\newcommand{\READING}

\ifdefined\READING
\documentclass[12pt,a4paper, onecolumn]{IEEEtran}
\else
\documentclass[10pt,a4paper, twocolumn]{IEEEtran}
\fi

\usepackage{amsmath, amsfonts, amsthm, amssymb}
\usepackage{graphicx}
\usepackage{subfigure}
\usepackage{hyperref}
\usepackage[small]{caption}

\ifdefined\READING
\usepackage{setspace}
\fi

\ifdefined\READING
\newcommand{\journaldtImageWidth}{120mm}
\newcommand{\journaldtImageHeight}{70mm}
\else
\newcommand{\journaldtImageWidth}{80mm}
\newcommand{\journaldtImageHeight}{50mm}
\fi

\hypersetup{
    colorlinks=false,
    pdfborder={0 0 0},
}

\theoremstyle{plain}
\newtheorem{theorem}{Theorem}[section]
\newtheorem{lemma}[theorem]{Lemma}
\newtheorem{proposition}[theorem]{Proposition}

\newtheorem{definition}[theorem]{Definition}
\newtheorem{remark}[theorem]{Remark}

\newcommand{\ceiling}[1]{\left\lceil{#1}\right\rceil}
\newcommand{\floor}[1]{\left\lfloor{#1}\right\rfloor}
\newcommand{\nfrac}[2]{\left( \frac{#1}{#2} \right)}
\newcommand{\brac}[1]{\left\{ #1 \right\}}
\newcommand{\bras}[1]{\left[ #1 \right]}
\newcommand{\brap}[1]{\left( #1 \right)}
\newcommand{\fpow}[3]{\left(\frac{#1}{#2}\right)^{#3}}
\newcommand{\dep}{\frac{\delta\epsilon_{a}}{\delta\epsilon_{a} + \epsilon_{V}}}
\newcommand{\depn}{\frac{\delta\epsilon_{a}}{\delta\epsilon_{a} + d}}

\newcommand{\Exp}{\mathbb{E}}
\newcommand{\Expp}{\mathbb{E}_{\pi}}

\newcommand{\sZ}{\mathbb{Z}_{+}}
\newcommand{\sR}{\mathbb{R}_{+}}
\newcommand{\ExpH}{\mathbb{E}_{\pi_{H}}}

\newcommand{\Expsqh}{\mathbb{E}_{S \vert Q, H}}
\newcommand{\Deq}{\stackrel{\Delta} = }
\newcommand{\Qg}{\overline{Q}(\gamma)}
\newcommand{\Pg}{\overline{P}(\gamma)}
\newcommand{\Qgq}{\overline{Q}(\gamma, q_{0})}
\newcommand{\Pgq}{\overline{P}(\gamma, q_{0})}

\newcommand{\Ag}{\overline{A}(\gamma)}
\newcommand{\Sg}{\overline{S}(\gamma)}
\newcommand{\Qgk}{\overline{Q}(\gamma_{k})}
\newcommand{\Pgk}{\overline{P}(\gamma_{k})}
\newcommand{\Ugk}{\overline{U}(\gamma_{k})}
\newcommand{\Agk}{\overline{A}(\gamma_{k})}
\newcommand{\Sgk}{\overline{S}(\gamma_{k})}

\newcommand{\sq}{\overline{s}(q)}
\newcommand{\sQ}{\overline{s}(Q)}

\DeclareMathOperator*{\mini}{minimize}

\newcommand*{\qeda}{\hfill\ensuremath{\blacksquare}}

\begin{document}

\title{On the tradeoff of average delay \\ and average power for fading point-to-point links \\ with monotone policies}

\author{Vineeth B. S. and Utpal Mukherji \\ Department of Electrical Communication Engineering, \\Indian Institute of Science, Bangalore - 560012.\\\{vineeth,utpal\}@ece.iisc.ernet.in}

\maketitle

\begin{abstract}
We consider a fading point-to-point link with packets arriving randomly at rate $\lambda$ per slot to the transmitter queue.
We assume that the transmitter can control the number of packets served in a slot by varying the transmit power for the slot.
We restrict to transmitter scheduling policies that are monotone and stationary, i.e., the number of packets served is a non-decreasing function of the queue length at the beginning of the slot for every slot fade state.
For such policies, we obtain asymptotic lower bounds for the minimum average delay of the packets, when average transmitter power is a small positive quantity $V$ more than the minimum average power required for transmitter queue stability.
We show that the minimum average delay grows either to a finite value or as $\Omega\brap{\log(1/V)}$ or $\Omega\brap{1/V}$ when $V \downarrow 0$, for certain sets of values of $\lambda$.
These sets are determined by the distribution of fading gain, the maximum number of packets which can be transmitted in a slot, and the transmit power function of the fading gain and the number of packets transmitted that is assumed.
We identify a case where the above behaviour of the tradeoff differs from that obtained from a previously considered approximate model, in which the random queue length process is assumed to evolve on the non-negative real line, and the transmit power function is strictly convex.
We also consider a fading point-to-point link, where the transmitter, in addition to controlling the number of packets served, can also control the number of packets admitted in every slot.
We obtain asymptotic lower bounds for the minimum average delay of the packets under a constraint on the average throughput, when average transmitter power is a small positive quantity $V$ more than the minimum average power required for transmitter queue stability.
We show that the minimum average delay grows either to a finite value or as $\Omega\brap{\log(1/V)}$ when $V \downarrow 0$, for certain sets of values of $\lambda$.
Our approach, which uses bounds on the stationary probability distribution of the queue length, also leads to an intuitive explanation of the asymptotic behaviour of average delay in the regime where $V \downarrow 0$.
\end{abstract}

\section{Introduction}

We study the optimal tradeoff of average delay with average power for point-to-point communication links with random arrivals and fading.
In this paper, we obtain asymptotic lower bounds on the minimum average delay in the asymptotic regime where average power is made arbitrarily close to the minimum average power which is required for stability of the transmitter queue.
Such asymptotic lower bounds are significant since they can be used to quantify how \emph{close} the performance of scheduling and power control policies are to the optimal.

We consider a slotted time queueing model with random arrival of packets.
The packets are assumed to be buffered in an infinite length queue.
The random fade state is assumed to be constant in each slot.
In each slot, the transmitter schedules a batch, with say $s$ packets, to be transmitted over the point-to-point fading link.
We assume that the transmitter expends a power of $P(h, s)$ watts when transmitting $s$ packets when the fade state is $h$.
For a particular policy of operation, the performance measures that we consider are the time average power and the time average queue length.
We note that for cases of interest, the average delay can be obtained from the time average queue length.
In this paper, we consider the tradeoff of average queue length with the average power over all possible policies.
Equivalently, our objective is to characterize the minimum average queue length subject to a constraint on the average power over all possible policies.
We mainly consider two queueing models, I-model and R-model, for which we characterize the tradeoff.
For I-model, we assume that the queue length evolution is on the set of non-negative integers, while for R-model, we assume that the queue length evolution is on the set of non-negative real numbers, with $P(h, s)$ being a strictly convex function in $s$.
We note that R-model is usually used as a tractable approximation to I-model.

We note that interesting solutions to the tradeoff problem have finite average queue length.
It is intuitive that if the average queue length has to be finite, the time-average service rate has to be equal to the average arrival rate $\lambda$.
Therefore, there is a minimum positive average power that is expended for the average service rate to be equal to $\lambda$.
It turns out that this minimum average power is a function of $\lambda$ for a given fade distribution, maximum batch service size $S_{max}$, and power cost function $P(h, s)$.
For the following discussion, this minimum average power is denoted as $c(\lambda)$ for I-model and $c_{R}(\lambda)$ for R-model.
In this paper, we consider the tradeoff problem in an asymptotic regime, denoted $\Re$, in which the difference, $V$, between the average power constraint and $c(\lambda)$ (or $c_{R}(\lambda)$ for R-model) is made arbitrarily close to zero.
We obtain and compare the asymptotic lower bounds on the minimum average delay in the asymptotic regime $\Re$ for both I-model and R-model.
We show that R-model with the strictly convex $P(h, s)$ function is not an appropriate approximation for I-model.

\subsection{Related work}
In this section, we review existing results for the tradeoff problem, on the basis of the approach used to obtain such results.
The most common approach which has been used to address the tradeoff problem has been to formulate it as a constrained Markov decision problem (CMDP), as in \cite{sennott} or \cite{altman}.

The CMDP is further analyzed by considering an equivalent Markov decision problem (MDP) which is obtained via a Lagrange relaxation \cite{ma}.
The MDP for models which are quite similar to I-model has been studied by Berry and Gallager \cite{berry}, Collins and Cruz \cite{collins}, and Goyal et al. \cite{munish}.
Agarwal et al. \cite{agarwal} studies the MDP for R-model.
In all these papers, the authors show that a stationary deterministic optimal policy exists, which is independent of the initial queue length.
This policy prescribes an optimal batch size to be used for service, when the queue length is $q$ and the fade state is $h$, which is a non-decreasing function of $q$ for every fade state $h$.
This monotonicity property of any stationary deterministic optimal policy motivates the definition of admissible policies in this paper.

Another approach has been to obtain an asymptotic order characterization of the minimum average queue length in the asymptotic regime $\Re$ where $V \downarrow 0$ for R-model.
Berry and Gallager \cite{berry} obtained an asymptotic lower bound for the minimum average queue length in the regime $\Re$ for the R-model with a strictly convex transmit power function.
Neely \cite{neely_mac} extended this asymptotic lower bound to a multiuser downlink model, where the transmit power function is a strictly convex function of the vector of fade states and the vector of service batch sizes for the users.
Other asymptotic bounds were also obtained by Neely in \cite{neely_mac}, \cite{neely_utility}, and \cite{huang_neely}.
Extensions to more general networks and other performance measures can be found in \cite{neely}.
A summary of these results is given in Table \ref{table:asymptotic_bounds}.
\begin{table*}[th!]
  \centering
  \begin{tabular}{|l|l|c|c|}
    \hline
     & 
    \textbf{Model details} & 
    \begin{minipage}{0.25\textwidth}
      \textbf{Asymptotic upper bound (Regime $\Re$)}
    \end{minipage} &
    \begin{minipage}{0.25\textwidth}
      \textbf{Asymptotic lower bound (Regime $\Re$)}
    \end{minipage} \\
    \hline
    1 &
    \begin{minipage}{0.4\textwidth}
      \vspace{0.2em}
      Berry-Gallager power delay tradeoff \cite{berry}; R-model with $S_{max} = \infty$
      \vspace{0.2em}
    \end{minipage} & $\mathcal{O}\brap{\frac{1}{\sqrt{V}}\log\nfrac{1}{V}}$ & $\Omega\nfrac{1}{\sqrt{V}}$ \\
    \hline
    2 &
    \begin{minipage}{0.4\textwidth}
      \vspace{0.2em}
      Multiuser Berry-Gallager power delay tradeoff \cite{neely_mac}; Multiuser R-model
      \vspace{0.2em}
    \end{minipage} & $\mathcal{O}\brap{\frac{1}{\sqrt{V}}\log\nfrac{1}{V}}$ & $\Omega\nfrac{1}{\sqrt{V}}$ \\
    \hline
    3 &
    \begin{minipage}{0.4\textwidth}
      \vspace{0.2em}
      Multiuser Berry-Gallager power delay tradeoff \cite{neely_mac}; Multiuser R-model, but with piecewise linear $P(h,s)$ and $\lambda$ such that $c_{R}(\lambda)$ is on a piecewise linear portion of $c_{R}(.)$
      \vspace{0.2em}
    \end{minipage} & $\mathcal{O}\brap{\log\nfrac{1}{V}}$ &
    \begin{minipage}{0.25\textwidth}
      \vspace{0.2em}
      $\Omega\brap{\log\nfrac{1}{V}}$ shown for a specific example, not known in general
      \vspace{0.2em}
    \end{minipage} \\
    \hline
    4 &
    \begin{minipage}{0.4\textwidth}
      \vspace{0.2em}
      Multiuser Berry-Gallager power delay tradeoff \cite{neely_mac}; Multiuser R-model, but with piecewise linear $P(h,s)$, $\lambda$ is  any abscissa at which the slope of $c_{R}(.)$ changes
      \vspace{0.2em}
    \end{minipage} & $\mathcal{O}\nfrac{1}{V}$ &
    \begin{minipage}{0.2\textwidth}
      \begin{center}
        Not known
      \end{center}
    \end{minipage} \\
    \hline
    5 &
    \begin{minipage}{0.4\textwidth}
      \vspace{0.2em}
      Power delay tradeoff with lower bound constraint on average throughput \cite{neely_utility}
      \vspace{0.2em}
    \end{minipage} & $\mathcal{O}\brap{\log\nfrac{1}{V}}$ &
    \begin{minipage}{0.25\textwidth}
      $\Omega\brap{\log\nfrac{1}{V}}$ but with single fade state
    \end{minipage} \\
    \hline
    6 &
    Utility delay tradeoff \cite{neely_superfast} & $\mathcal{O}\brap{\log\nfrac{1}{V}}$ & $\Omega\brap{\log\nfrac{1}{V}}$ \\
    \hline
    7 & 
    \begin{minipage}{0.4\textwidth}
      \vspace{0.2em}
      Power delay tradeoff with Markov arrival and fading process \cite{huang_neely}
      \vspace{0.2em}
    \end{minipage} & $\mathcal{O}\nfrac{1}{V}$ & 
    \begin{minipage}{0.2\textwidth}
      \begin{center}
        Not known
      \end{center}
    \end{minipage} \\    
    \hline
  \end{tabular}
  \caption{Available asymptotic lower bounds (with upper bounds) on the minimum average queue length; except for case 7 all other models assume that the arrival process and the fade process are IID, and except for cases 5 and 6 all models do not have admission control. Also, all lower bounds are derived under the assumption that the queue length can take real values.}
  \label{table:asymptotic_bounds}
\end{table*}
We note that asymptotic lower bounds are not known in many cases.
Such asymptotic lower bounds are significant, since they may help in determining the best possible tradeoff.
We also note that known asymptotic lower bounds have been derived under the assumption that the queue length evolution is on $\sR$.

Monotonicity properties of optimal policies have also been obtained for continuous time queueing models in \cite{weber}, \cite{george}, \cite{ata_pcstatic}, and \cite{atamm1}.
Order optimality has also been explored for finite buffer systems in \cite[Chapter 6]{berry_thesis}.
Asymptotic order bounds have been obtained in a variety of other cases also, as in \cite{venkialtman}, \cite{chaporkar_alex}, and \cite{botan}.

\subsection{Overview of the paper}
In this paper, we (i) obtain asymptotic lower bounds for cases, for which they are not known, for a restricted class of monotone stationary policies called admissible policies, (ii) obtain an intuitive explanation of the behaviour of the tradeoff in the regime $\Re$ via bounds on the stationary probability of the queue length, and (iii) compare the behaviour of the tradeoff for I-model and R-model, showing that R-model with the strictly convex power cost function may be inappropriate as an approximate model.
We first formulate the tradeoff problem as a constrained Markov decision problem (CMDP) in Section \ref{sec:cmdp_formulation}, wherein we show that there exists an optimal stationary policy.
Then we consider a Markov decision process (MDP) obtained from a Lagrangian relaxation of the above CMDP in Section \ref{sec:mdp_formulation}.
We review the structural properties of the optimal policy for the above MDP in Section \ref{sec:mdp_formulation}.
These structural properties are then used to motivate the definition of admissible policies in Section \ref{sec:admissible_policies_definition}.

We formulate the tradeoff problem for admissible policies in Section \ref{sec:tradeoffproblem_admissible_policies}.
We then define the asymptotic regime $\Re$ for I-model and R-model rigorously in Section \ref{sec:imodelrmodel_asymp_regime}.
In Section \ref{sec:asymptotic_prelims}, preliminary results which lead to the asymptotic lower bounds are discussed.
We present bounds on the stationary probability distribution for admissible policies in Section \ref{sec:upperbounds}.
These bounds are useful in obtaining an intuitive explanation of the behaviour of the tradeoff for admissible policies in the regime $\Re$.

In Section \ref{sec:asymp_analysis} we first present an intuitive explanation of the behaviour of the tradeoff for admissible policies in the regime $\Re$ and then obtain asymptotic lower bounds for the minimum average delay.
We note that an overview of results in this paper was presented earlier in \cite{vineeth_ncc13_dt}.
In this paper, we provide a complete analysis, as well as extend our results to a queueing model with admission control and an additional constraint on the utility of average throughput. 
We study the queueing models with admission control (I-model-U and R-model-U) in Section \ref{sec:systemmodel_modelus}.
We conclude the paper in Section \ref{sec:conclusions}.

\subsection{Notation and conventions}
We use the following notation for the asymptotic bounds: (i) $f(x)$ is $\mathcal{O}(g(x))$ if there exists a $c > 0$ such that $\lim_{x \rightarrow 0} \frac{f(x)}{g(x)} \leq c$; $f(x), g(x) \geq 0$, (ii) $f(x)$ is $\Omega(g(x))$ if there exists a $c > 0$ such that $\lim_{x \rightarrow 0} \frac{f(x)}{g(x)} \geq c$; $f(x), g(x) \geq 0$, and (iii) $f(x)$ is $\omega(g(x))$ if $\lim_{x \rightarrow 0} \frac{f(x)}{g(x)} = \infty$.
All logarithms are natural logarithms unless specified otherwise.
Sequences which are monotonically increasing to a limit point are denoted as $\uparrow$, while those monotonically decreasing are denoted as $\downarrow$.
The stationary version of a random process is denoted by dropping the time index, e.g. $Q \sim Q[m]$.
We denote the set of non-negative integers and non-negative real numbers by $\sZ$ and $\sR$ respectively.
In this paper, we use the following notation for the limits of integration, for integrals of functions with respect to probability measures: $\int_{a^{+}}^{b^{-}} f(x) dP(x) = \int_{-\infty}^{\infty} f(x)\mathbb{I}\brac{a < x < b} dP(x)$.

\section{System model}
\subsection{System model - Integer valued queue length evolution}
\label{sysmodel:intval}
We consider a discrete time system with slots indexed by $m \in \brac{1, 2, \cdots}$, as in Figure \ref{fig:system}.
In each slot $m$, a random number of packets, $R[m] \in \sZ$, where each packet is of the same size, arrive into the transmitter queue.
The arrival sequence $(R[m], m \geq 1)$ is assumed to be independently and identically distributed (IID) with $R[1] \leq A_{max}$, batch arrival rate $\mathbb{E}R[1] = \lambda < \infty$, $var(R[1]) = \sigma^{2} < \infty$.
We initially consider a case without admission control.
Then the number of packets admitted $A[m] = R[m]$.
The packets are assumed to arrive into an infinite buffer, in which they wait until they are transmitted over a point to point fading channel.
\begin{figure}
  \centering
  \includegraphics[width=\journaldtImageWidth,height=\journaldtImageHeight]{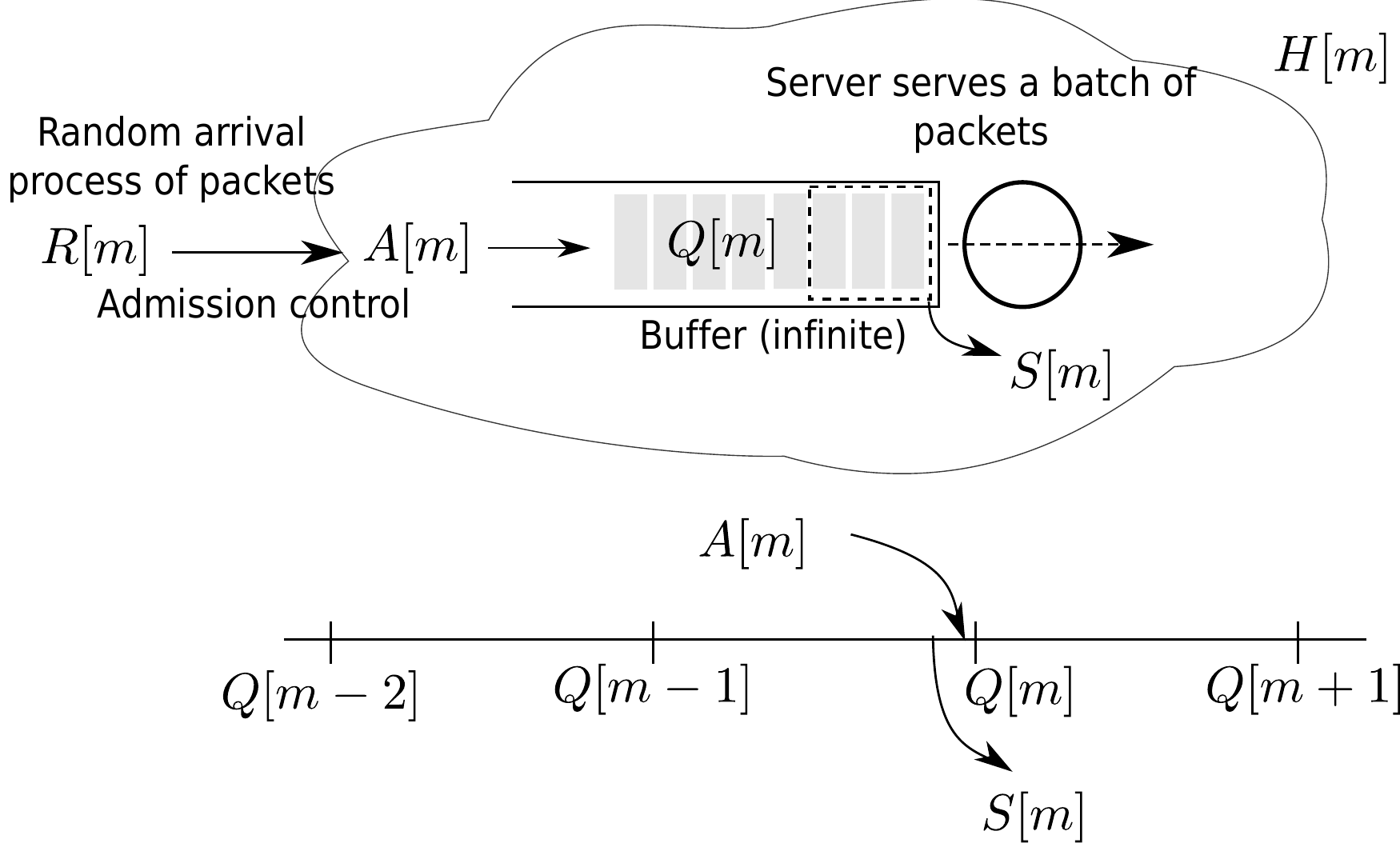}
  \caption{The discrete time single server queueing model with a single queue}
  \label{fig:system}
\end{figure}

The fade state is assumed to be constant in a slot.
The fade state takes values in a finite set $\mathcal{H}$, with $\min\brac{\mathcal{H}} > 0$, and the fade state process $(H[m], m \geq 1)$, is assumed to be IID, with $H[1] \sim \pi_{H}$.
The expectation with respect to $\pi_{H}$ is denoted by $\Exp_{\pi_H}$.
The processes $(R[m])$ and $(H[m])$ are assumed to be independent of each other.

The number of customers in the queue at the start of the $(m + 1)^{th}$ slot is denoted by $Q[m]$.
The system is assumed to start with $Q[0] = q_{0} \in \mathbb{Z}_{+}$ customers.
At the end of slot $m$, a batch with $S[m] \in \sZ$ packets is removed from the transmitter queue just before the $A[m]$ new packets which arrive in the $m^{th}$ slot are admitted.
We assume that $S[m] \leq \min\brap{Q[m - 1], S_{max}}$, where $S_{max}$ is the maximum batch size that can be served.
The queue evolution sampled at the slots is given by:
\begin{equation}
  Q[m + 1] = Q[m] - S[m + 1] + A[m + 1].
  \label{eq:evolution}
\end{equation}
The evolution of the queue length is also illustrated in \mbox{Figure \ref{fig:system}}.

At the start of slot $m$, the history of the system is defined as $\sigma[m] \Deq (q_{0}, H[1], S[1], Q[1], H[2], S[2], Q[2], \dots,$ \mbox{$Q[m - 2]$}, \mbox{$H[m - 1]$}).
At the beginning of slot $m$, the transmitter scheduler observes  $H[m]$ and chooses a batch service size $S[m] \in \sZ$ as a randomized function of the history $\sigma[m]$, the current queue length $Q[m - 1]$, and the current fade state $H[m]$.
We define a policy $\gamma$ to be the sequence of such mappings (service batch sizes) $(S[1], S[2], \dots)$.
The set of all policies is denoted by $\Gamma$.
If $\gamma$ is such that $S[m] = S(Q[m - 1], H[m])$, where $S(q, h)$ is a randomized function, then $\gamma$ is a stationary policy.
The set of all stationary policies is denoted as $\Gamma_{s}$.
We note that since $H[m]$ is assumed to be IID, if $\gamma \in \Gamma_{s}$ then the process $(Q[m], m \geq 0)$ is a discrete time Markov chain (DTMC).

When the fade state is $h$, the transmitter expends $P(h, s)$ units of power when transmitting $s$ bits.
We note that $P(h, s)$ is a function of the fading gain $h^{2}$, when the fade state is $h$.
Motivated by many examples (see \cite{berry} and \cite{elif}) of the form for $P(h,s)$, we assume that $\forall h \in \mathcal{H}$, $P(h,s)$ satisfies the following properties:
\begin{description}
\item[C1 : ]$P(h,0) = 0$, and 
\item[C2 : ]$P(h,s)$ is non-decreasing and convex in $s$, for $s \in \brac{0, \dots, S_{max}}$, for every $h \in \mathcal{H}$.
\end{description}

The average power for $\gamma \in \Gamma$ is
\begin{equation}
  \overline{P}(\gamma, q_{0}) \stackrel{\Delta} = \limsup_{M \rightarrow \infty} \frac{1}{M} \Exp \bras{\sum_{m = 1}^{M} P(H[m], S[m]) \middle \vert Q[0] = q_{0}}.
  \label{chap5fading:eq:avgpower}
\end{equation}
The average queue length for $\gamma \in \Gamma$ is
\begin{equation}
  \overline{Q}(\gamma, q_{0}) \stackrel{\Delta} = \limsup_{M \rightarrow \infty} \frac{1}{M} \Exp \bras{\sum_{m = 0}^{M - 1} Q[m] \middle \vert Q[0] = q_{0}}.
  \label{chap5fading:eq:avgqlength}
\end{equation}
We consider the optimal tradeoff of $\overline{P}(\gamma, q_{0})$ with $\overline{Q}(\gamma, q_{0})$ for this model.
Using Little's law, the optimal tradeoff between $\overline{P}(\gamma, q_{0})$ and average delay can then be obtained.

\subsection{System model - Real valued queue length evolution}
\label{sysmodel:realval}
In this section, we describe a queueing model, which is usually used as an analytically tractable approximation for the model discussed above.

We state only the differences from the model discussed in the previous section.
We assume that for $m \geq 1$, $A[m] \in [0, A_{max}]$, $S[m] \in [0, S_{max}]$, and $q_{0} \in \mathbb{R}_+$.
Hence, the queue length $Q[m] \in \mathbb{R}_+, \forall m \geq 0$.
The function $P(h,s)$ is assumed to satisfy the following properties:
\begin{description}
\item[RC1 : ]{$P(h,0) = 0,$ for every $h \in \mathcal{H}$,}
\item[RC2 : ]{$P(h,s)$ is non-decreasing and strictly convex in $s$, for $s \in [0, S_{max}]$, for every $h \in \mathcal{H}$.}
\end{description}
The average power and average queue length are defined as in \eqref{chap5fading:eq:avgpower} and \eqref{chap5fading:eq:avgqlength} respectively.

We note that this model is similar to that considered by Berry and Gallager \cite{berry}, except that in \cite{berry}, $S_{max} = \infty$.

In the following, the model in the previous section, where the queue length evolution was assumed to be on the non-negative integers, is called the I-model, while the model described here, where the queue length evolution is assumed to be on the non-negative real numbers, is called the R-model.
We note that R-model with fading and $P(h,s)$ being strictly convex is usually used as an analytically tractable approximation for I-model\footnote{The strict convexity of $P(h,s)$ helps in algebraic manipulations, which are used in obtaining approximations for value functions and optimal policies, such as in \cite{munish}, \cite{bettesh}, and \cite{chen}.}.
When used as an approximation, $P(h,s)$ for R-model coincides with $P(h,s)$ for I-model for $s \in \brac{0, \dots, S_{max}}, \forall h$.
We note that there are also scenarios where it is natural to model the queue evolution as evolving on real numbers, with a strictly convex cost function, such as when the queue is assumed to buffer a certain amount of error exponent as in \cite{berry}.

\subsection{The tradeoff problem}
\label{sec:problem}
Our objective is to characterize the minimum average queue length for a given constraint $P_{c}$ on the average transmit power.
We note that the following statement is for both the I-model and the R-model.
The tradeoff problem is:
\begin{equation}
  \mini_{\gamma \in \Gamma} \Qgq, \text{ such that } \Pgq \leq P_{c}.
  \label{eq:tradeoffproblem_init}
\end{equation}
If an optimal policy exists for the above problem, then it is denoted as $\gamma^*(q_{0}, P_{c})$.
We note that one of the ways in which the above constrained optimization problem can be solved is by considering its Lagrange dual, the dual function of which is as follows:
\begin{equation}
  \mini_{\gamma \in \Gamma} \brac{ \Qgq + \beta \bras{\Pgq - P_{c}} }.
  \label{eq:tradeoffproblem_init_dual}
\end{equation}
where $\beta \geq 0$ can be interpreted as a Lagrange multiplier.
If an optimal policy exists for the above problem, then it is denoted as $\gamma^*_{\beta}(q_{0})$.

\subsection{An example }
\label{sec:simpleeg}
Throughout this paper, to illustrate the results for I-model and R-model, we use the following example.
We assume that packets of size $100$-bits arrive in a random process $(A[m])$ to the transmitter queue.
The number of packets $A[1]$ which arrive in a slot is assumed to be distributed according to a Binomial$(A_{max},p)$ distribution, with arrival rate $r_{a} kb/s$.
The rate of service $r$ in $kb/s$ is assumed to be $200\log_{10}\brap{1 + SNR}$, where $SNR$ is the received signal to noise ratio.
We assume that $SNR = \frac{h^{2} P}{L}$, where $h^{2}$ is the fading gain, $P$ is the transmit power, and $L$ encompasses the loss due to attenuation as well as noise power.

We assume that the slots are of duration $2ms$.
Then the arrival rate of packets in a slot is $\lambda = \frac{r_{a}}{50}$.
We assume that if $h^{2} = 1$ and $P = 1W$ then $r = 50$.
Therefore, if $P(h, r)$ is the transmit power as a function of the fade state and the rate, we have that $P(h, r) = \frac{1.28}{h^{2}}\brap{10^{r/200} - 1}$.
We note that in one slot, the number of bits served is $2r$.
We assume that the transmitter, in each slot, can choose its transmission rate in the set $\brac{0, 50, 100} kb/s$.
To fit this example to our model, we express the queue length in units of $100$ bits.
Then in each slot, we have a Binomial arrival process of $100$-bit packets and service of $s$ $100$-bit packets, where $s \in \brac{0, 1, 2}$.
The transmit power as a function of $h$ and $s$ is $P(h,s) = \frac{1.28}{h^{2}}\brap{10^{50s/200} - 1}, s \in \brac{0, 1, 2}$.
We note that the average queue length, as defined, is in units of $100$ bits.

For the R-model, the set of possible batch sizes is assumed to be $[0, 2]$.
The transmit power as a function of $h$ and $s$ is assumed to be $P(h,s) = \frac{1.28}{h^{2}}\brap{10^{50s/200} - 1}$ but for $s \in [0, 2]$.

\section{Problem formulation}
\subsection{Formulation as a CMDP}
\label{sec:cmdp_formulation}

In this section, we formulate \eqref{eq:tradeoffproblem_init} as a CMDP with an average cost objective and an average cost constraint to conclude that it is sufficient to consider \eqref{eq:tradeoffproblem_init} for the class of stationary policies.

The state space $\mathcal{X}$ of the CMDP is $\sZ \times \mathcal{H}$ for I-model and $\sR \times \mathcal{H}$ for R-model.
The action space at each $(q,h) \in \mathcal{X}$ is the discrete set $\brac{0,\dots,\min(q, S_{max})}$ for I-model and the interval $[0, \min(q,S_{max})]$ for R-model.
The probabilistic evolution of the process is as given in \eqref{eq:evolution} for both I-model and R-model.
Associated with the CMDP, there are two single stage costs: (i) the queue length cost $q$, and \mbox{(ii) the} power cost $P(h, s)$, where $s$ is the chosen service batch size.
We note the above single stage costs correspond to the average cost objective and the average cost constraint respectively.

For \eqref{eq:tradeoffproblem_init}, if (i) $\lambda < S_{max}$ and (ii) the constraint $P_{c}$ is such that $P_{c} > c(\lambda)$, then from \cite{hernandez} and \cite{hernandez_2}, it is possible to show that \eqref{eq:tradeoffproblem_init} has an optimal solution and there exists an optimal policy $\gamma \in \Gamma_{s}$ for I-model.
For the R-model, it can similarly be shown that if $\lambda < S_{max}$ and $P_{c} > c_{R}(\lambda)$, then 
\eqref{eq:tradeoffproblem_init} has an optimal solution and there exists an optimal policy $\gamma \in \Gamma_{s}$.
Hence, in the following, we restrict to $\gamma \in \Gamma_{s}$ for studying \eqref{eq:tradeoffproblem_init}.
In the following, we assume that $\lambda < S_{max}$, $P_{c} > c(\lambda)$ for I-model, and $P_{c} > c_{R}(\lambda)$ for R-model.

\subsection{Formulation as a MDP}
\label{sec:mdp_formulation}
We note that \eqref{eq:tradeoffproblem_init_dual} is an MDP with average cost criterion.
We consider the I-model first.
The state space of the MDP is $\sZ \times \mathcal{H}$.
The action space at each $(q, h) \in \sZ \times \mathcal{H}$ is $\brac{0, \dots, \min(q, S_{max})}$.
The probabilistic evolution of the process is as in \eqref{eq:evolution}.
The single state cost associated with the MDP is $q + \beta P(h, s)$, where $\beta \geq 0$ and $s$ is the action.

It is shown (e.g. \cite{berry}, \cite{munish}) that if $\lambda < S_{max}$, then there exists an optimal solution for \eqref{eq:tradeoffproblem_init_dual}, and there exists an optimal stationary deterministic policy $\gamma^*_{\beta}(q_{0})$.
Let $g^*_{\beta}(q_{0})$ be the optimal value of \eqref{eq:tradeoffproblem_init_dual}.

We now make the following assumptions about the arrival process $(A[m])$: (A1) $Pr\brac{A[1] > S_{max}} > \epsilon_{a} > 0$, and (A2) $Pr\brac{A[1] = a} > 0$, for all $a \in \brac{0,\dots, A_{max}}$.
The analysis of the MDP leads to the following observations: (O1) the optimal average cost $g^*_{\beta} = g^*_{\beta}(q_{0})$ and any stationary deterministic optimal policy $\gamma^*_{\beta} = \gamma^*_{\beta}(q_{0})$ are independent of the initial state $q_{0}$, and $g^*_{\beta} < \infty$;
(O2) any stationary deterministic optimal policy $\gamma^*_{\beta}$, which serves $s^*_{\beta}(q, h)$ in state $(q, h)$, is such that $s^*_{\beta}(q,h)$ is non-decreasing in $q$, $\forall h \in \mathcal{H}$;
(O3) For any policy, from assumptions A1 and A2, we note that from state 0 it is possible to reach any other state $q$.
From O2, we obtain that any stationary deterministic optimal policy has a single recurrence class $\mathcal{R}_{\gamma^*_{\beta}}$, of the form $\brac{q_{m},\dots}$, where $q_{m} = \min\brac{q : \exists q' > q, Pr\brac{Q[m + 1] = q \vert Q[m] = q'} > 0}$ for the optimal policy under consideration;
(O4) We note that $s^*_{\beta}(q_{m}, h) = 0$ by definition. From A2, we have that $q_{m}$ is an aperiodic state, and therefore the class $\mathcal{R}_{\gamma^*_{\beta}}$ of the Markov chain under $\gamma^*_{\beta}$ is aperiodic.

We note that both O1 and O2 have been obtained in \cite{berry_thesis} and \cite{munish}.
For R-model, the state space of the MDP is $\sR \times \mathcal{H}$.
The action space at each $(q, h) \in \sR \times \mathcal{H}$ is $[0,S_{max}]$.
We note that O2 has been obtained for R-model in \cite{agarwal}.
The above observations are used to motivate the definition of admissible policies in this paper.

\subsection{Correspondence between the solutions of \eqref{eq:tradeoffproblem_init} and \eqref{eq:tradeoffproblem_init_dual}}
\begin{definition}[The set $\mathcal{O}^{u}$]
  Let $\Gamma^*_{\beta}$ be the set of all optimal policies for \eqref{eq:tradeoffproblem_init_dual} for a $\beta \geq 0$.
  Let $P_{c}(\beta) \Deq \brac{ \Pg, \gamma \in \Gamma^*_{\beta} }$.
  Then, we define $\mathcal{O}^{u}$ as
  \begin{eqnarray*}
    \mathcal{O}^{u} \Deq \bigcup_{\beta \geq 0} P_{c}(\beta)
  \end{eqnarray*}
\end{definition}
We note that $\mathcal{O}^{u}$ is the set of average power values for optimal policies for \eqref{eq:tradeoffproblem_init_dual} for all $\beta \geq 0$.

\begin{definition}[The set $\mathcal{O}_{d}^{u}$]
  Let $\Gamma^*_{d,\beta}$ be the set of all stationary deterministic optimal policies for \eqref{eq:tradeoffproblem_init_dual} for a $\beta \geq 0$.
  Let $P_{c,d}(\beta) \Deq \brac{ \Pg, \gamma \in \Gamma^*_{d,\beta} }$.
  Then, we define $\mathcal{O}_{d}^{u}$ as
  \begin{eqnarray*}
    \mathcal{O}_{d}^{u} \Deq \bigcup_{\beta \geq 0} P_{c,d}(\beta)
  \end{eqnarray*}
\end{definition}
We note that $\mathcal{O}_{d}^{u}$ is the set of average power values for stationary deterministic optimal policies for \eqref{eq:tradeoffproblem_init_dual} for all $\beta \geq 0$.

From Ma et al. \cite{ma}, if the constraint $P_{c}$ in \eqref{eq:tradeoffproblem_init} is such that $P_{c} \in \mathcal{O}^{u}$, then the solutions of \eqref{eq:tradeoffproblem_init} and \eqref{eq:tradeoffproblem_init_dual} coincide.
Furthermore, if $P_{c} \in \mathcal{O}_{d}^{u}$, and if A1 and A2 hold, then there exists a stationary deterministic optimal policy for \eqref{eq:tradeoffproblem_init} for I-model which possesses the properties O1, O2, O3, and O4.
For R-model, if $P_{c} \in \mathcal{O}_{d}^{u}$, then there exists a stationary deterministic policy for \eqref{eq:tradeoffproblem_init}, which possesses the properties O1 and O2.

\subsection{Admissible policies}
\label{sec:admissible_policies_definition}
In this section, we define the set of admissible policies, for which we obtain an asymptotic characterization of the tradeoff in the regime $\Re$.
The following definitions are for both the I-model as well as the R-model.
\begin{definition}[Stable policies]
A stationary policy $\gamma$ is \emph{stable} if: (i) the Markov chain $(Q[m], m \geq 0)$ under $\gamma$ is irreducible, aperiodic, and positive Harris recurrent with stationary distribution $\pi$, and (ii) $\overline{Q}(\gamma, q_{0}) < \infty$.
\end{definition}

\begin{definition}[Admissible policies]
A policy $\gamma$ is \emph{admissible} if: (G1) it is stable, and, (G2) the average service rate in state $q$, $\overline{s}(q) \stackrel{\Delta} = \Exp_{\pi_{H}} \Exp_{S|q,H} S(q, H)$ is non-decreasing in $q$.
The set of admissible policies is denoted as $\Gamma_{a}$.
\end{definition}

We note that the stability property of admissible policies has been motivated by O1, O3, and O4.
It can be shown that (\cite[Lemma 4.3.2]{vineeth_thesis}) if $|\mathcal{H}| = 1$ then any stationary admissible policy for \eqref{eq:tradeoffproblem_init_dual} is non-idling, i.e., $s^*_{\beta}(q,h) > 0$ for $q > 0$.
This motivates the assumption of irreducibility.
Property G2 has been motivated by O2.
We note that whenever $P_{c} \in \mathcal{O}_{d}^{u}$, there exists an optimal admissible policy for \eqref{eq:tradeoffproblem_init}.
For an admissible policy $\gamma$, we have that the performance measures $\overline{Q}(\gamma, q_{0})$ and $\overline{P}(\gamma, q_{0})$ are independent of the initial queue state $q_{0}$ and exist as limits.
We note that our definition of admissible policies includes G2 in addition to the properties of admissible policies in \cite{berry}.
Therefore, in the following, these performance measures are denoted by $\Qg$ and $\Pg$ respectively.

\subsection{Tradeoff problem for admissible policies}
\label{sec:tradeoffproblem_admissible_policies}
The TRADEOFF problem for admissible policies is:
\begin{equation}
  \mini_{\gamma \in \Gamma_{a}} \overline{Q}(\gamma), \text{ such that } \overline{P}(\gamma) \leq P_{c}.
  \label{chap5fading:eq:tradeoffprob}
\end{equation}
The optimal value of TRADEOFF is denoted as $Q^*(P_{c})$.

For an admissible policy $\gamma$, we note that since the arrival rate is constant, from Little's law the average delay for $\gamma$ is $\frac{\overline{Q}(\gamma)}{\lambda}$.
The minimum average delay as a function of the average power constraint $P_{c}$ for admissible policies is $\frac{Q^*(P_{c})}{\lambda}$.

We note that the tradeoff problem can be formulated for a larger class of policies $\Gamma_{a,M}$, which is obtained by mixing or time sharing of policies in $\Gamma_{a}$.
Let $Q^*_{M}(P_{c})$ denote the optimal value of \eqref{chap5fading:eq:tradeoffprob}, but with the minimization carried out over the set $\Gamma_{a,M}$.
The asymptotic behaviour for $Q^*_{M}(P_{c})$ can be directly obtained from $Q^*(P_{c})$ (e.g. \cite[Proposition 2.3.9]{vineeth_thesis}).
Therefore, in the following we consider the asymptotic characterization of $Q^*(P_{c})$ only.

We note that there may not exist an optimal admissible policy for TRADEOFF.
But by definition, we have that there exists admissible policies which are arbitrarily \emph{good} when $P_{c}$ is such that \mbox{TRADEOFF} is feasible.
\begin{definition}[$\epsilon$-optimal policies]
If $P_{c}$ is such that \mbox{TRADEOFF} is feasible, then by definition there exists a feasible admissible policy $\gamma$ such that $\Qg \leq Q^*(P_{c}) + \epsilon$.
Such a policy is defined to be $\epsilon$-optimal.
\end{definition}

\subsection{Asymptotic regime $\Re$}
\label{sec:imodelrmodel_asymp_regime}
We note that for any admissible policy $\gamma$, $\overline{Q}(\gamma) = \Exp_{\pi}Q$ and $\overline{P}(\gamma) = \Expp\Exp_{H | Q}  \Exp_{S|Q,H} P(H, S(Q,H))$.
Since $Q$ and $H$ are independent, we further have that $\overline{P}(\gamma) = \Expp\Exp_{\pi_{H}}  \Exp_{S|Q,H} P(H, S(Q,H))$.
For any $\gamma \in \Gamma_{a}$, we note that the average arrival rate $\lambda$ is  equal to the average service rate, i.e., $\lambda = \Expp \Exp_{H|Q}  \Exp_{S|Q,H} S(Q,H) = \Expp \overline{s}(Q)$.

Therefore, for $\gamma \in \Gamma_{a}$, $\overline{P}(\gamma)$ is lower bounded by the optimal value of
\begin{eqnarray}
  \mini_{\gamma \in \Gamma_{a}} & \Expp \Exp_{{H}|Q}  \Exp_{S|Q,H} P(H, S(Q, H)), \nonumber \\
  \text{such that } & \Expp \Exp_{{H}|Q} \Exp_{S|Q,H} S(Q, H) = \lambda,
  \label{eq:tradeoffprobratec}
\end{eqnarray}
since the only constraint is on the average service rate.
We note that $\Expp \Exp_{{H}|Q} \Exp_{S|Q,H} S(Q, H) = \Exp_{\pi_{H}} \Exp_{Q|H} \Exp_{S|Q,H} S(Q, H)$.
Then, we have that 
\[\Exp_{Q|H} \Exp_{S|Q,H} S(Q, H) = \int_{q} \int_{s} s. dp_{s|q,H}. d\pi(q),\]
\[ = \int_{s} \int_{q} s .dp_{s,q|H} = \int_{s} s \int_{q} dp_{s,q|H},\] 
we have that $\Exp_{\pi_{H}} \Exp_{Q|H} \Exp_{S|Q,H} S(Q, H) = \Exp_{\pi_{H}} \Exp_{S|H} S$ where the conditional distribution of $S$ given $H$ depends upon the policy.
A similar procedure can be carried out on $\Expp \Exp_{H|Q}  \Exp_{S|Q,H} P(H, S(Q, H))$ which leads to $\Expp \Exp_{H|Q}  \Exp_{S|Q,H} P(H, S(Q, H)) = \Exp_{\pi_{H}} \Exp_{S|H} P(H, S)$.
Then the optimal value of \eqref{eq:tradeoffprobratec} is bounded below by the optimal value of 
\begin{eqnarray}
  \mini & \Exp_{\pi_{H}}  \Exp_{S|H} P(H, S),
  \label{eq:app_bg_0} \\
  \text{such that } & \Exp_{\pi_{H}}  \Exp_{S|H} S = \lambda, \nonumber
\end{eqnarray}
where we minimize over all possible conditional distributions for the batch size $S$ given $H$, irrespective of the policy.
For the I-model, we denote the optimal value of \eqref{eq:app_bg_0} by $c(\lambda)$, while for the R-model we denote the optimal value of \eqref{eq:app_bg_0} by $c_{R}(\lambda)$\footnote{We note that $c(\lambda)$ and $c_{R}(\lambda)$ are the minimum average powers required for mean rate stability for the I-model and R-model respectively, see \cite{neely_mac}.}.
We note that for the R-model, the conditional distribution of the batch size has support on $[0, S_{max}]$, while for the I-model the conditional distribution has support on $\brac{0, \dots, S_{max}}$.
Hence, $c_{R}(\lambda) \leq c(\lambda), \forall \lambda \in [0, S_{max}]$.
We note that feasible solutions exist for the above problem only if $\lambda \leq S_{max}$.

We have that $\forall \gamma \in \Gamma_{a}$, $\overline{P}(\gamma) \geq c(\lambda)$ for the I-model, and $\overline{P}(\gamma) \geq c_{R}(\lambda)$ for the R-model.
From \cite[Theorem 1]{neely_mac}, we have that if $\lambda < S_{max}$, then for a sequence $P_{c,k} \downarrow c(\lambda)$ (or $P_{c,k} \downarrow c_{R}(\lambda)$ for the R-model), there exists a sequence of admissible policies $(\gamma_{P_{c,k}})$, such that $\overline{P}(\gamma_{P_{c,k}}) \leq P_{c,k}$, and $\overline{Q}(\gamma_{P_{c,k}})$ grows without bound.
Since, for an arrival rate of $\lambda$, $c(\lambda)$ (or $c_{R}(\lambda)$ for the R-model) can be approached arbitrarily closely by admissible policies, $c(\lambda)$ (or $c_{R}(\lambda)$ for the R-model) is the infimum of $\Pg, \gamma \in \Gamma_{a}$.

\begin{definition}[The regime $\Re$]
The asymptotic regime $\Re$ for TRADEOFF is defined as the regime in which $P_{c} \downarrow c(\lambda)$ for the I-model and $P_{c} \downarrow c_{R}(\lambda)$ for the R-model.
\end{definition}

\section{Asymptotic analysis - Preliminaries}
\label{sec:asymptotic_prelims}
\subsection{Properties of $c(\lambda)$ and $c_{R}(\lambda)$}

For I-model, since properties (C1) and (C2) are assumed to hold, from \cite[Section VII]{neely_mac}, we have that $c(\lambda)$ is a piecewise linear, non-decreasing convex function, for $\lambda \in [0, S_{max}]$, with $c(0) = 0$.
Again from \cite{neely_mac}, $c_{R}(\lambda)$ is a non-decreasing, strictly convex function of $\lambda \in [0, S_{max}]$, with $c_{R}(0) = 0$.
For the example discussed in Section \ref{sec:simpleeg}, the function $c(\lambda)$ and $c_{R}(\lambda)$ are illustrated in Figure \ref{fig:clambdaeg_2} with $\mathcal{H} = \brac{0.1, 1}$ and $\pi_{H}(0.1) = 0.6$.
In this case, $c_{R}(\lambda) < c(\lambda)$ for all $\lambda \not \in \brac{0.4, 0.8, 1.4}$.
Let $\Lambda$ be the set of $\lambda$ at which the slope of $c(\lambda)$ changes.
We note that $c_{R}(\lambda) \leq c(\lambda)$.
From similar examples, we have observed that $c_{R}(\lambda) = c(\lambda)$ for $\lambda \in \Lambda$ and $c_{R}(\lambda) < c(\lambda)$ for $\lambda \not \in \Lambda$.

\begin{figure}[h]
  \centering
  \includegraphics[width=\journaldtImageWidth,height=\journaldtImageHeight]{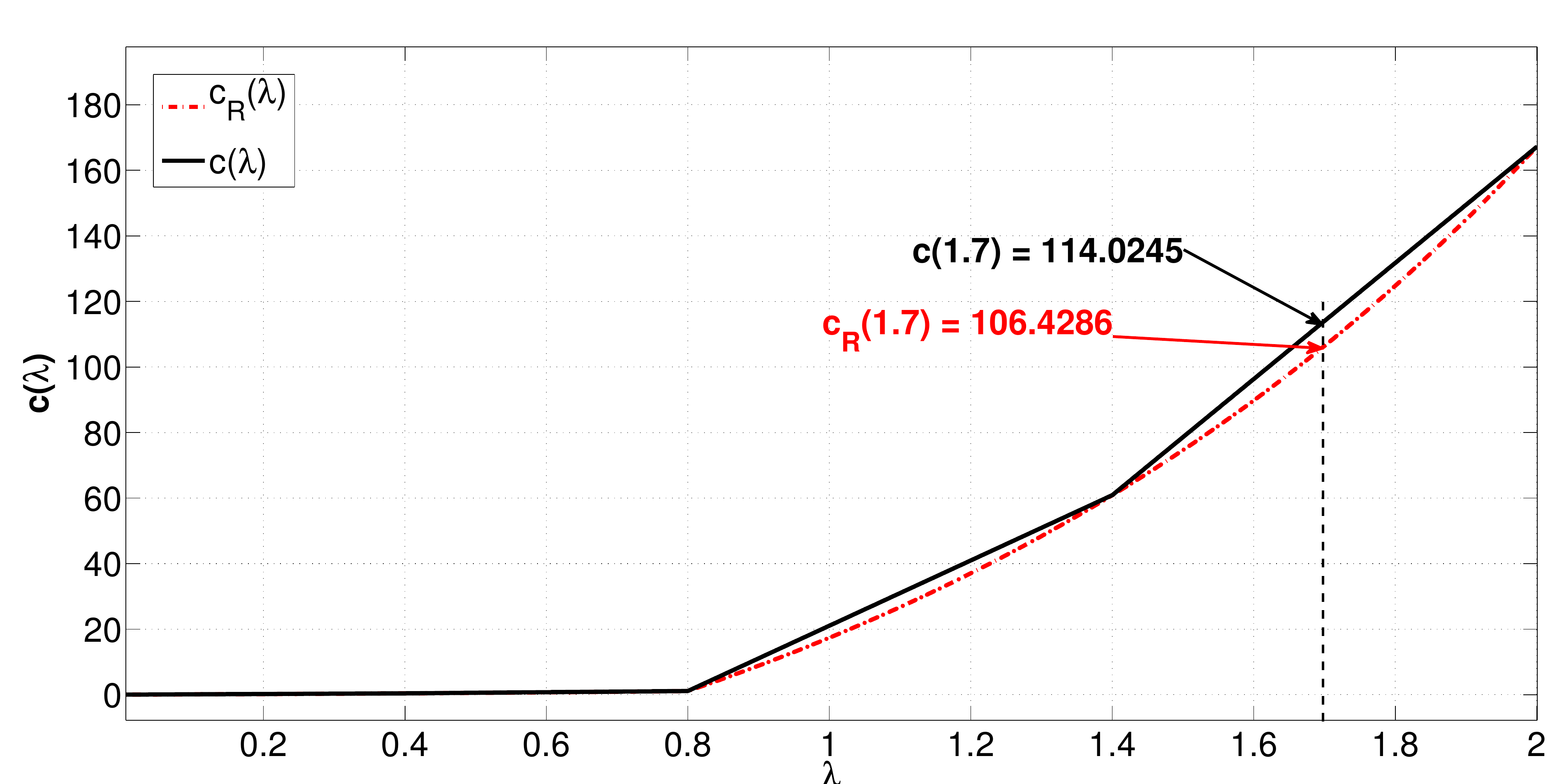}
  \caption{\small{The optimal value of problem \eqref{eq:app_bg_0}: $c(\lambda)$ for I-model and $c_{R}(\lambda)$ for R-model, with $H\in \brac{0.1, 1}$ and $\pi_{H}(0.1) = 0.6$; $c(1.7)$ is $107\%$ of $c_{R}(1.7)$.}}
  \label{fig:clambdaeg_2}
\end{figure}

We will observe from the asymptotic analysis, that the asymptotic growth rate of $Q^*(P_{c})$ in the regime $\Re$, suggested by the R-model and the I-model for $\lambda \in \Lambda$ are different.
We note that, in general, the set $\Lambda$ is not known analytically.
However, if $|\mathcal{H}| = 1$, it is clear that $\Lambda = \brac{0, \dots, S_{max}}$, from the construction of $P(h,s)$ for the R-model.

\subsection{Cases for I-model and R-model}
\label{sec:differentcases}
The behaviour of ${Q}^*(P_{c})$ as $P_{c} \downarrow c(\lambda)$ is observed to depend on the form of $c(s), s \in [0, S_{max}]$ in the \emph{neighbourhood} of $\lambda$.
Since $c(s)$ is piecewise linear, we can define a sequence of intervals $[a_{p}, b_{p}]$, $p \in \brac{1,\dots, P}$, with $a_{p + 1} = b_{p}$, $a_{1} = 0$, and $b_{P} = S_{max}$.
The sequence of intervals is such that for $s \in [a_{p}, b_{p}]$, $c(s)$ is linear.
The following three cases arise:
\begin{enumerate}
\item $a_{1} = 0 < \lambda < b_{1}$,
\item $a_{p} < \lambda < b_{p}$, $p > 1$, and,
\item $\lambda = a_{p}, p > 1$.
\end{enumerate}
We note that $c_{R}(a_{p}) = c(a_{p})$ in Figure \ref{fig:clambdaeg_2} and $\Lambda = \brac{a_{p}, p \in \brac{1, \dots, P}}$.
\begin{figure*}[thbp!]
  \centering
  \includegraphics[width=160mm,height=35mm]{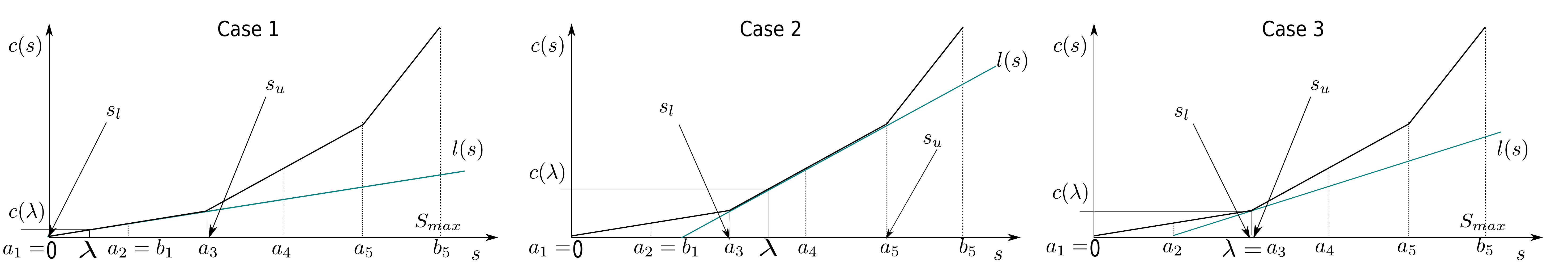}
  \caption{\small{Illustration of $c(s)$ and $l(s)$, along with the relationship between $\lambda, s_{l}$, and $s_{u}$ for the three cases for I-model.}}
  \label{fig:threecases}
\end{figure*}

For Cases 1 and 2, let $s_{l} \stackrel{\Delta} = a_{p}$ and $s_{u} \stackrel{\Delta} = b_{p}$, while for Case 3 let $s_{l} = s_{u} \stackrel{\Delta} = a_{p}$.
An example is shown in Figure \ref{fig:threecases}.

We define the line $l(s) : [0,S_{max}] \rightarrow \mathbb{R}_+$ as follows:
\begin{enumerate}
\item If $s_{l} < \lambda < s_{u}$, then $l(s)$ is the line through $(s_{l},c(s_{l}))$ and $(s_{u},c(s_{u}))$.
\item If $s_{l} = \lambda = s_{u} = a_{p}$ for some $p > 1$, then $l(s)$ is a line through $(\lambda, c(\lambda))$ with slope $m$ chosen such that $\frac{c(a_{p}) - c(a_{p - 1})}{a_{p} - a_{p - 1}} < m < \frac{c(a_{p + 1}) - c(a_{p})}{a_{p + 1} - a_{p}}$.
\end{enumerate}
We note that $\Expp l(\overline{s}(Q)) = c(\lambda)$.

For the R-model, since $c_{R}(s)$ is strictly convex in $s \in [0, S_{max}]$, we consider the cases where $\lambda$ is such that $c_{R}(s)$ has a positive second derivative at $s = \lambda$.

\subsection{A lower bound which is dependent on $\gamma$, $\gamma \in \Gamma_{a}$}
\label{sec:gammadep_lowerbound}

Unlike $c(\lambda)$ ($c_{R}(\lambda)$), in this section, we obtain a lower bound on $\Pg$ for I-model (for R-model) which is dependent on $\gamma$, for any $\gamma \in \Gamma_{a}$.
This lower bound will prove to be useful in the derivation of the asymptotic lower bounds.

We consider the I-model first.
We note that the average power used when the queue length is $q$ is $\Exp_{\pi_{H}}  \Exp_{S|q,H} P(H, S(q, H))$, which is bounded below by the optimal value of
\begin{eqnarray*}
  \mini & & \Exp_{\pi_{H}}  \Exp_{S|H} P(H, S), \\
  \text{such that } & & \Exp_{\pi_{H}}  \Exp_{S|H} S = \bar{s}(q),
\end{eqnarray*}
where we have considered all possible conditional distributions on the batch size with support on $\brac{0, \dots, S_{max}}$, subject only to the constraint that the average service rate is $\bar{s}(q)$.
The above optimization problem is the same as \eqref{eq:app_bg_0} except that the constraint is now $\bar{s}(q)$ instead of $\lambda$.
Therefore, the average power used when the queue length is $q$ is bounded below by $c(\overline{s}(q))$.
Then for that $\gamma$, $\Expp c(\overline{s}(Q)) \leq \overline{P}(\gamma)$.

We note that any feasible policy $\gamma$ for TRADEOFF has $\overline{P}(\gamma) \leq P_{c}$.
Therefore, $\Expp c(\overline{s}(Q)) \leq P_{c}$.
We also note that from the convexity of $c(s)$, $\Expp c(\overline{s}(Q)) \geq c(\lambda)$, since $\Expp \overline{s}(Q) = \lambda$.
Now as $P_{c} \downarrow c(\lambda)$ in the regime $\Re$, for any sequence of feasible policies for TRADEOFF, $\Exp c(\overline{s}(Q)) \downarrow c(\lambda)$.

Similar to the I-model, it can be shown that the average power used when the queue length is $q$ is bounded below by $c_{R}(\overline{s}(q))$ for R-model.
Then any feasible policy $\gamma$ for TRADEOFF has $\Expp c_{R}(\overline{s}(Q)) \leq \overline{P}(\gamma)$.
Since $\gamma$ is feasible, $\overline{P}(\gamma) \leq P_{c}$, and $\Expp c_{R}(\overline{s}(Q)) \leq P_{c}$.
We also note that from the convexity of $c_{R}(s)$, $\Expp c_{R}(\overline{s}(Q)) \geq c_{R}(\lambda)$, since $\Expp \overline{s}(Q) = \lambda$.
Now as $P_{c} \downarrow c_{R}(\lambda)$, for any sequence of feasible policies for TRADEOFF, $\Exp c_{R}(\overline{s}(Q)) \downarrow c_{R}(\lambda)$.

\subsection{Bounds on the stationary probability for admissible policies}
\label{sec:upperbounds}
In this section, we first present two bounds on the stationary probability $\pi$ for $\gamma \in \Gamma_{a}$.
For R-model, we assume : (RA1) $Pr\brac{A[1] - S_{max} > \delta_{a}} > \epsilon_{a}$, for positive $\delta_{a}$ and $\epsilon_{a}$. 
We note that RA1 is similar to A1.

The first bound that we derive is obtained via properties of the transition probability of the DTMC $(Q[m])$ for $\gamma \in \Gamma_{a}$.
\begin{proposition}
  Let $\pi$ denote the stationary probability distribution for $(Q[m])$ under $\gamma \in \Gamma_{a}$.
  Let $Q \sim \pi$.
  We assume that $Pr\brac{A[1] = 0} > 0$ for both I-model and R-model.
  Let $q_{1} = \sup\brac{q : \overline{s}(q) < s_{1}}$, where $0 < s_{1} < S_{max}$.
  For R-model, for any $\Delta$ and $k$ such that $0 < \Delta <s_{1}$ and $k \geq 0$, we have that 
  \small
  \begin{eqnarray*}
     Pr\brac{Q \in [q_{1} + k \Delta, q_{1} + (k + 1) \Delta)} \leq Pr\brac{Q < q_{1}}\frac{ \brap{1 + \frac{1}{\rho_{d}}}^{k}}{\rho_{d}},
  \end{eqnarray*}
  \normalsize
  where $\rho_{d} \stackrel{\Delta} = \nfrac{s_{1} - \Delta}{S_{max} - \Delta}Pr\brac{A[1] = 0}$.
  For I-model, for any $k \geq 0$ and $q = q_{1} + k$,
  \begin{equation*}
    \pi(q) \leq Pr\brac{Q < q_{1}} \frac{\left(1 + \frac{1}{\rho_{d}}\right)^{k}}{\rho_{d}},
  \end{equation*}
  where $\rho_{d} \stackrel{\Delta} = \nfrac{s_{1}}{S_{max}}Pr\brac{A[1] = 0}$.
  \label{prelimresults:txprob_real}
\end{proposition}
The proof is given in Appendix \ref{app:prelimresults:txprob_real}.
For R-model, another upper bound on $\pi$, which uses a weaker assumption, $Pr\brac{A[1] \leq \epsilon} > 0$, with $\epsilon > 0$, is presented in \cite[Lemma 5.7.2]{vineeth_thesis}.
The second bound that we derive below is obtained using a Lyapunov drift based method.
This bound is an extension of the geometric lower bound on the stationary probability of DTMCs presented in \cite{gamarnik_1} to the cases with state dependent drift ($\lambda - \sq$) and state space being $\sR$.

\begin{proposition}
  Let $\pi$ be the stationary probability distribution for $(Q[m])$ under $\gamma \in \Gamma_{a}$.
  Let $Q \sim \pi$.
  Suppose there exists a finite $q_{d}$ such that \[ \forall q, 0 \leq q \leq q_{d}, \mathbb{E}\bras{Q[m + 1] - Q[m] \middle \vert Q[m] = q} \geq -d,\] where $d$ is positive.
  Let $\epsilon_{a}$ and $\delta_{a}$ be as in assumption RA1.
  For R-model, for any $q_{1}$, $k \geq 0$, $\Delta > 0$, $\delta > 0$, $\Delta + \delta < \delta_{a}$, and $0 \leq q_{1} + k\Delta \leq q_{d}$, we have 
  \begin{eqnarray*}
    & Pr\brac{Q \geq q_{1} + k\Delta} \geq \brap{\depn}^{k} Pr\brac{Q \geq q_{1}} + \\
    & \brap{1 - \brap{\depn}^{k}} \bras{ Pr\brac{Q > q_{d}} - \frac{1}{d}\int_{q_{d}^{+}}^{\infty} (\sq - \lambda)d\pi(q)}.
  \end{eqnarray*}
  Let $\epsilon_{a}$ be as defined in assumption A1.
  For I-model, for any $q_{1}$, $k \geq 0$ such that $0 \leq q_{1} + k \leq q_{d}$, we have
  \begin{eqnarray*}
    & Pr\brac{Q \geq q_{1} + k } \geq \brap{\frac{\epsilon_{a}}{\epsilon_{a} + d}}^{k} Pr\brac{Q \geq q_{1}} + \\
    & \brap{1 - \brap{\frac{\epsilon_{a}}{\epsilon_{a} + d}}^{k}}\bras{Pr\brac{Q > q_{d}} - \frac{1}{d}\sum_{q_{d} + 1}^{\infty} (\sq - \lambda)\pi(q)}.
  \end{eqnarray*}
  \label{prelimresults:drift_real}
\end{proposition}
The proof is given in Appendix \ref{app:prelimresults:drift_real}.

We note that the bounds on $\pi(q)$ derived above, are dependent on either the stationary probability $Pr\brac{Q < q_{1}}$ or  the \emph{tail-drift} $\sum_{q = q_{d} + 1}^{\infty} \brap{\lambda - \sq}\pi(q)$.
We note that $\Qg = \sum_{q = 0}^{\infty} q d\pi(q)$.
In order to obtain bounds on $\Qg$ in terms of $\Pg$, we relate $\pi(q)$ to $\Pg$.
One of ways in which this can be done is to relate $Pr\brac{Q < q_{1}}$ or $\sum_{q = q_{d} + 1}^{\infty} \brap{\sq - \lambda}\pi(q)$ to $\Pg$.
In the following lemma, we obtain upper bounds on $Pr\brac{Q < q_{1}}$ as a function of $\Pg$.
The following result will also turn out to be useful in providing intuition for the asymptotic behaviour of $Q^*(P_{c})$ in the regime $\Re$.
\begin{lemma}
Let $\pi$ be the stationary probability distribution for $(Q[m])$ under $\gamma \in \Gamma_{a}$.
Let $Q \sim \pi$.
Let $V \Deq \Pg - c(\lambda)$ for I-model, and $V \Deq \Pg - c_{R}(\lambda)$ for R-model.
Let $\mathcal{S} \subseteq [0, S_{max}]$, and $\mathcal{Q}_{S} \Deq \brac{q : \sq \in \mathcal{S}}$.
Let $\epsilon_{V}$ be positive.
Then for the I-model, we have
\begin{eqnarray*}
  Pr\brac{Q \in \mathcal{Q}_{S}} = Pr\brac{\sQ \in \mathcal{S}} \leq \frac{V}{m\epsilon_{V}},
\end{eqnarray*}
where $\mathcal{S} \subseteq [0, s_{l} - \epsilon_{V}) \bigcup (s_{u} + \epsilon_{V}, S_{max}]$.
For the R-model, $\exists a > 0$, such that
\begin{eqnarray*}
  Pr\brac{Q \in \mathcal{Q}_{S}} = Pr\brac{\sQ \in \mathcal{S}} \leq \frac{V}{a\epsilon^{2}_{V}},
\end{eqnarray*}
where $\mathcal{S} \subseteq [0, \lambda - \epsilon_{V}) \bigcup (\lambda + \epsilon_{V}, S_{max}]$.
\label{lemma:prelimresults:serviceratebound}
\end{lemma}
The proof of this lemma is given in Appendix \ref{app:prelimresults:serviceratebound}.
We note that in the asymptotic regime $\Re$, $\epsilon_{V}$ is usually chosen as a sequence depending on $V$.
By choosing $\epsilon_{V}$ to be $\omega(V)$ and $\omega(\sqrt{V})$ for I-model and R-model respectively, the above bounds are observed to be decreasing to zero, as $V \downarrow 0$.
We note that an upper bound on $Pr\brac{Q < q_{1}}$ as a function of $\Pg$ can be obtained by choosing $q_{1}$ to be such that $\overline{s}(q_{1}) \in [0, s_{l} - \epsilon_{V}]$ for both I-model and R-model.

\section{Asymptotic analysis}
\label{sec:asymp_analysis}
In this section, we present intuitive explanations for the behaviour of $Q^*(P_{c})$ for I-model and R-model.
Then, we present and compare the asymptotic lower bounds on $Q^*(P_{c})$ for I-model and R-model.

We note that the behaviour of $Q^*(P_{c})$ for Case 1 is different from all other cases.
Hence, we consider Case 1 separately.
For Case 1, we have analytical results only for the case where $|\mathcal{H}| = 1$.
Then, we have that $c(s)$ is the lower convex envelope of $P(h_{0}, s)$, where $h_{0}$ is the single fade state, with $a_{p} \in \Lambda = \brac{0, \dots, S_{max}}$.
For Case 1, from Lemma \ref{lemma:prelimresults:serviceratebound}, for any sequence of feasible $\gamma_{k} \in \Gamma_{a}$, we have that $Pr\brac{\sQ > s_{u} + \epsilon_{V}} \downarrow 0$, for $\epsilon_{V} = \sqrt{V}$ and $V \Deq P_{c} - c(\lambda)$.
Therefore, any $\gamma \in \Gamma_{a}$ with $\Pg = c(\lambda)$ has $Pr\brac{S[m] > s_{u}} = 0, \forall m \geq 1$.

Let the policy $\gamma_{u}$ be such that $S[m + 1] = \min(Q[m], s_{u})$.
For any given realization of the arrival process and randomization of batch sizes, let $q^*[m]$ and $q[m]$ be the evolution of the queue process under $\gamma_{u}$ and another admissible policy $\gamma$ respectively, with $\Pg = c(\lambda)$.
Then we note that $q^*[m] \leq q[m], \forall m$, and therefore $\gamma_{u}$ has the least average queue length over all $\gamma \in \Gamma_{a}$ such that  $\Pg = c(\lambda)$.
Since $\overline{Q}(\gamma_{u})$ is finite, for Case 1 $Q^*(P_{c})$ does not grow to infinity as $P_{c}$ approaches $c(\lambda)$.
For $V > 0$, we note that service rates which are greater than $s_{u}$ could be used, i.e., for all $q \in \brac{0, \dots, q_{u}}$, $\sq \leq s_{u}$, while $\sq > s_{u}$ for $q > q_{u}$.
Intuitively, this implies that $Q^*(P_{c}) < \overline{Q}(\gamma_{u})$ for $P_{c} > c(\lambda)$.
We note that $P_{c} - c(\lambda)$ would constrain $q_{u}$, and as $P_{c} \downarrow c(\lambda)$, $q_{u} \uparrow \infty$.
Instead of finding the minimum average queue length over all admissible policies with $\Pg \leq P_{c}$, we find the minimum average queue length over all policies with a particular value of $q_{u}$, which is constrained by $P_{c}$.
In the following, we present an asymptotic lower bound to $Q^*(P_{c})$ which is strictly less than $\overline{Q}(\gamma_{u})$ for $P_{c} > c(\lambda)$ and approaches $\overline{Q}(\gamma_{u})$ as $P_{c} \downarrow c(\lambda)$, using the above constraint relaxation.

We now discuss the behaviour of $Q^*(P_{c})$ for Cases 2 and 3 for I-model and for R-model, in the regime $\Re$.
For all these cases, we find that $Q^*(P_{c})$ increases to infinity in the regime $\Re$ for admissible policies.
From Lemma \ref{lemma:prelimresults:serviceratebound}, it can be shown that $Pr\brac{\sq \leq \epsilon} = \mathcal{O}(V)$, for a small positive $\epsilon$.
This implies that $Pr\brac{Q \leq \epsilon} = \mathcal{O}(V)$, i.e., the probability of the queue being (almost) empty goes to zero as $V \downarrow 0$.
Since the queue length process is a Markov chain for admissible policies, we expect that the stationary probability of any queue length which is reached from any $q \leq \epsilon$ in any finite number of slots (which does not depend on $V$) would also decrease to zero.
Then, as $V \downarrow 0$, one would expect that the stationary probability of any finite queue length should decrease to zero.
Therefore, intuitively, the average queue length has to increase.

The exact nature of $Q^*(P_{c})$ for the different cases depends on the \emph{shape} of $\pi(q)$ in the regime $\Re$.
Intuition about $\pi(q)$ is obtained by considering a simplified state dependent M/M/1 queueing model \cite[Chapters 2 and 3]{vineeth_thesis}.
The state dependent M/M/1 queueing model is a birth death process evolving on $\sZ$, with the state representing the queue length, controllable death rates $\mu(q), q \geq 1$ representing service rates, and controllable birth rates $\lambda(q), q \geq 0$ representing arrival rates.
We associate a cost rate function $\tilde{c}(\mu)$ with the service rate $\mu$ for this model, where $\tilde{c}(\mu)$ is the function $c(\mu)$ (or $c_{R}(\mu)$ if we need intuition about R-model).
We also note that $\mu(q) \in [0, S_{max}]$ and $\mu(0) = 0$.
Since for I-model and R-model, we do not have admission control, we choose $\lambda(q) = \lambda$.
We consider admissible policies $\Gamma_{a}$ for this simplified model, where a policy is the choice of $\mu(q)$.
An admissible policy is such that $\mu(q)$ is a non-decreasing function of $q$.

Let $\Delta(q) = \lambda - \mu(q)$ be the \emph{drift} in state $q$.
For admissible policies, since $\Delta(q)$ is a monotonically non-increasing function of $q$, with $\Delta(0) = \lambda > 0$ and $\lim_{q \uparrow \infty} \Delta(q) < 0$, the stationary probability distribution has the following \emph{shape} (as (S1) in Figure \ref{introduction:figure:intuition_2}). The stationary probability distribution $\pi(q)$ is monotonically increasing, then may or may not be constant for a set of queue lengths, and then is monotonically decreasing.
For I-model, let $\mathcal{Q}_{h} = \brac{q : \mu(q) \in [s_{l} - \epsilon_{V}, s_{u} + \epsilon_{V}]}$, where $\epsilon_{V}$ is $\omega(V)$ as $V \downarrow 0$.
Similarly for R-model, let $\mathcal{Q}_{h} = \brac{q : \mu(q) \in [\lambda - \epsilon_{V}, \lambda + \epsilon_{V}]}$.
Then proceeding as in Lemma \ref{lemma:prelimresults:serviceratebound} it can be shown that $Pr\brac{Q \in \mathcal{Q}_{h}} \uparrow 1$ as $V \downarrow 0$, i.e., as $V \downarrow 0$, the service rates have to be chosen from the set $[s_{l} - \epsilon_{V}, s_{u} + \epsilon_{V}]$ or $[\lambda - \epsilon_{V}, \lambda + \epsilon_{V}]$ for I-model and R-model respectively.

We consider Case 2 for I-model, we note that $s_{l} < s_{u}$, therefore as $V \downarrow 0$, the stationary probability distribution $\pi(q)$ is monotonically increasing (since $\lambda > s_{l}$), then may or may not be constant for a set of queue lengths, and then is monotonically decreasing (since $\lambda < s_{u}$).
Furthermore, if $q_{1} \Deq \min \mathcal{Q}_{h}$, then it can be shown that $\pi(q_{1} - 1) = \mathcal{O}(V)$.
We note that $\Qg \geq \frac{\overline{q}}{2}$, where $\overline{q} \in \mathcal{Q}_{h}$ is the largest queue length such that $Pr\brac{Q \leq \overline{q}} \leq \frac{1}{2}$.
We note that the stationary distribution increases geometrically from $\pi(q_{1} - 1) = \mathcal{O}(V)$, with the geometric factor being at least $\frac{\lambda}{s_{l}}$, i.e., $\pi(q) \leq \pi(q_{1} - 1) \fpow{\lambda}{s_{l}}{q - q_{1} + 1}$ for $q \in \mathcal{Q}_{h}$.
Then it can be shown that $\overline{q}$ is atleast $\Omega\brap{\log\nfrac{1}{V}}$.

We note that as $V \downarrow 0$, $\mu(q) \rightarrow \lambda$, for $q \in \mathcal{Q}_{h}$ for both Case 3 (I-model) and R-model.
Then, as $V \downarrow 0$, $\pi(q)$ is a constant for $q \in \mathcal{Q}_{h}$ ((S2) in Figure \ref{introduction:figure:intuition_2}).
It can be shown that $\pi(q_{1} - 1)$ is $\mathcal{O}(V)$ for Case 3 (I-model) and $\mathcal{O}(\sqrt{V})$ for R-model ((P1) and (P2) respectively in Figure \ref{introduction:figure:intuition_2}).
Then we note that $\pi(q)$ for $q \in \mathcal{Q}_{h}$ is equal to $\pi(q_{1} - 1)$.
Hence, $\overline{q}$, as defined above, is $\Omega\nfrac{1}{V}$ and $\Omega\nfrac{1}{\sqrt{V}}$ for Case 3 (I-model) and R-model respectively.
We have been unable to rigorously relate the distribution $\pi(q)$ for the simplified M/M/1 model discussed above to the discrete time queueing model considered in this paper.
However, the proofs for the asymptotic lower bounds for I-model and R-model are guided by the above intuition.

\begin{figure}
  \centering
  \includegraphics[width=\journaldtImageWidth,height=\journaldtImageHeight]{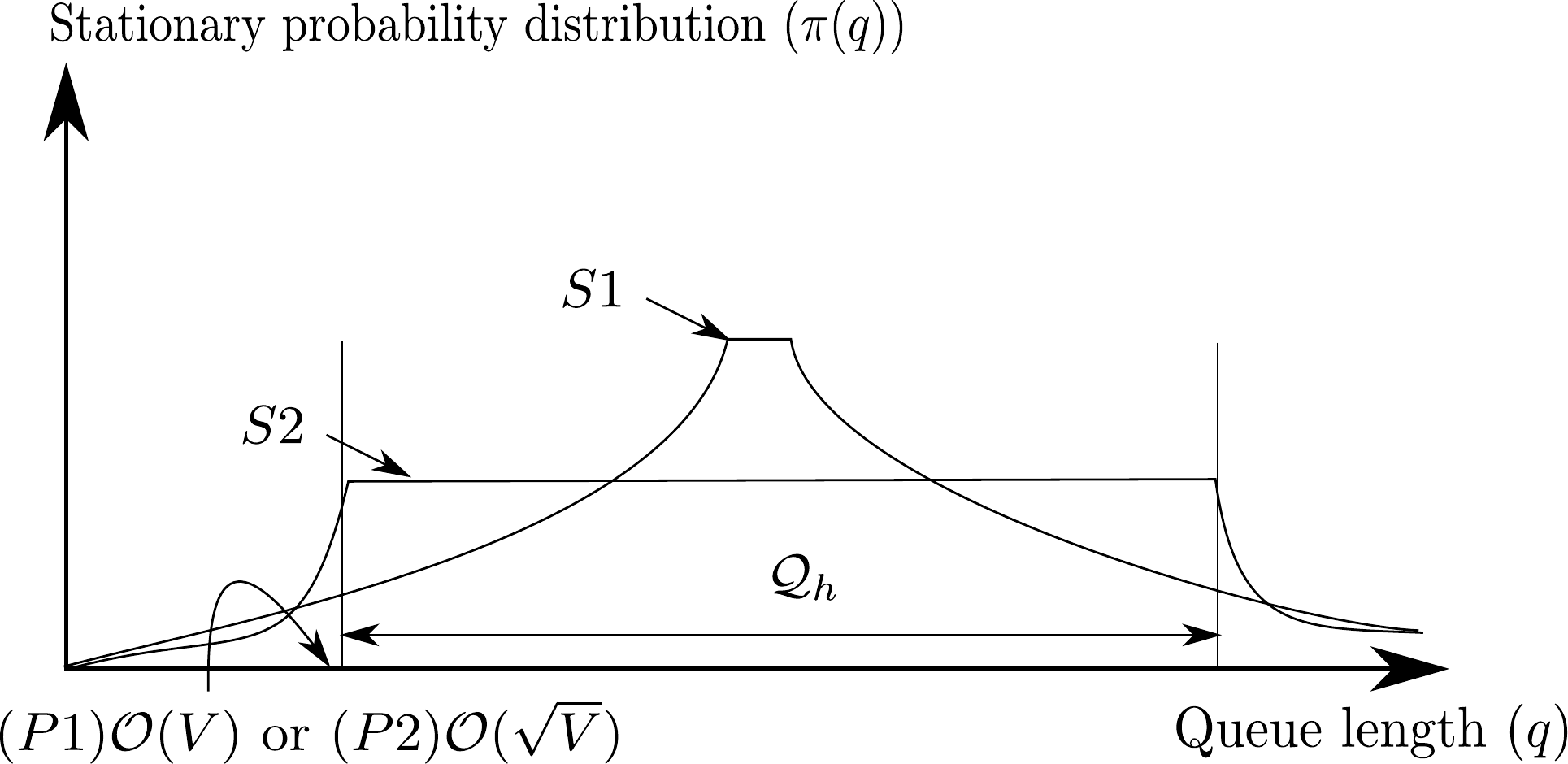}
  \caption{Possibilities for the behaviour of the stationary probability distribution in the regime $\Re$}
  \label{introduction:figure:intuition_2}
\end{figure}

\subsection{Asymptotic lower bounds for I-model}

We are able to obtain an analytical lower bound for Case 1 only under the additional assumption that $s_{u} = 1$ and by restricting to non-idling deterministic admissible policies.

\begin{proposition}
  For Case 1, with $|\mathcal{H}| = 1$, $s_{u} = 1$, and for any sequence of non-idling deterministic $\gamma_{k} \in \Gamma_{a}$, with $\Pgk - c(\lambda) = V_{k} \downarrow 0$, we have that $\overline{Q}(\gamma_{k}) = \frac{\sigma^{2}}{2(s_{u} - \lambda)} + \frac{\lambda}{2} - \mathcal{O}\brap{V_{k}\log\nfrac{1}{V_{k}}}$. Therefore, $Q^*(P_{c}) = \frac{\sigma^{2}}{2(s_{u} - \lambda)} + \frac{\lambda}{2} - \mathcal{O}\brap{\brap{P_{c} - c(\lambda)}\log\nfrac{1}{P_{c} - c(\lambda)}}$.
\end{proposition}
This result is derived using a series of steps, which are not presented here for brevity.
The complete derivation can be found in \cite[Lemmas 4.3.17, 4.3.18, 4.3.19]{vineeth_thesis}.
We note that if $s_{u} = 1$, then $\overline{Q}(\gamma_{u}) = \frac{\sigma^{2}}{2(s_{u} - \lambda)} + \frac{\lambda}{2}$, from \cite{denteneer}.
Thus as $V_{k} \downarrow 0$, we have that the asymptotic lower bound has $\overline{Q}(\gamma_{u})$ as the limit point.
We now consider Case 2.
\begin{proposition}
  For Case 2, given any sequence of admissible policies $\gamma_{k}$ with $\overline{P}(\gamma_{k}) - c(\lambda) = V_{k} \downarrow 0$, we have that $\overline{Q}(\gamma_{k}) = \Omega\brap{\log\nfrac{1}{V_{k}}}$. Therefore, $Q^*(P_{c}) = \Omega\brap{\log\nfrac{1}{P_{c} - c(\lambda)}}$.
  \label{lemma:case2}
\end{proposition}
The proof is given in Appendix \ref{appendix:proof:case2lb}.
The geometrically increasing bound on $\pi(q)$ from Proposition \ref{prelimresults:txprob_real} is used in deriving the above asymptotic lower bound.
We note that the geometric nature of the bound is similar to that which has been observed for the M/M/1 queueing model discussed above.
\begin{remark}
  We note that a similar asymptotic lower bound has been derived in \cite[Theorem 2]{neely_utility} (for a model with admission control) and in \cite[Theorem 2]{botan}.
  Assumption G3 has not been used in both papers.
  Although the above result has been derived independently, we note that underlying all the three derivations, there is the idea of bounding the probability of an event by a particular sequence of transitions for a Markov chain, i.e., a sequence of transitions in which the state of the Markov chain becomes successively smaller.
  Furthermore, in our proof, using assumption G3, we obtain geometric bounds on the stationary probability of any queue length, which is not available in \cite{neely_utility} as well as \cite{botan}.
\end{remark}

For Case 2, a tight asymptotic upper bound can be obtained from the sequence of Tradeoff Optimal Control Algorithm (TOCA) policies in \cite{neely_mac}.
However, TOCA policies are not admissible, since $S[m]$ depends on an auxiliary state variable, other than $Q[m - 1]$ and $H[m]$.
Therefore, in \cite[Chapter 5, Lemma 5.3.2]{vineeth_thesis}, we present a sequence of admissible policies that achieve the above asymptotic lower bound.

We now consider Case 3.
\begin{proposition}
  For Case 3, given any sequence of admissible policies $\gamma_{k}$ with $\overline{P}(\gamma_{k}) - c(\lambda) = V_{k} \downarrow 0$, we have that $\overline{Q}(\gamma_{k}) = \Omega\nfrac{1}{V_{k}}$. Therefore, $Q^*(P_{c}) = \Omega\nfrac{1}{P_{c} - c(\lambda)}$.
  \label{lemma:case3lb}
\end{proposition}
The proof is given in Appendix \ref{appendix:proof:case3lb}.

\begin{remark}
  The derivation of the relationship between the difference of the $\Pg$ and $c(\lambda)$ and the tail-drift defined as $\sum_{q =  q_{d} + 1}^{\infty} \bras{\lambda - \sq} \pi(q)$, in the proof of Proposition \ref{lemma:case3lb}, is motivated by the approach in \cite{berry}.
  We note that in our proof, $q_{d}$ can be chosen arbitrarily by the choice of $\epsilon_{V}$ and then $\epsilon_{V}$ can be chosen so as to obtain the tightest asymptotic lower bound.
  However, in \cite{berry}, $q_{d}$ cannot be chosen arbitrarily. 
  In fact, $q_{d}$ is the queue length which has the maximal stationary probability of all queue lengths in the set $\brac{0,\dots,\ceiling{2\Qg}}$ for a policy.
  The freedom in the choice of $q_{d}$ enables us to derive the $\Omega\nfrac{1}{V}$ asymptotic lower bound.
\end{remark}

For Case 3, a sequence of admissible policies can be obtained as in \cite{neely_mac}, which achieve the above asymptotic lower bound.

We note that the asymptotic lower bounds, derived above for admissible policies, hold for \eqref{eq:tradeoffproblem_init} for $P_{c} \in \mathcal{O}^{u}_{d}$.
We note that the optimal solution of \eqref{eq:tradeoffproblem_init_dual} is a lower bound for \eqref{eq:tradeoffproblem_init}.
In \cite[Proposition 4.3.22]{vineeth_thesis}, we obtain asymptotic lower bounds for the optimal solution of \eqref{eq:tradeoffproblem_init_dual} using the asymptotic lower bounds in Propositions \ref{lemma:case2} and \ref{lemma:case3lb}.
Then these asymptotic lower bounds for \eqref{eq:tradeoffproblem_init_dual} are used to obtain asymptotic lower bounds on \eqref{eq:tradeoffproblem_init}, for a set of $P_{c}$ which includes $\mathcal{O}^{u}_{d}$.
We note that the minimum average delay as a function of the power constraint $P_{c}$, can be obtained from Little's law as $\frac{Q^*(P_{c})}{\lambda}$.

\subsection{Numerical examples}

\label{sec:numerical}
To illustrate the results obtained in the previous section, we plot the optimal tradeoff curve for the example in Section \ref{sec:simpleeg}.
We assume that $A_{max} = 5$, $\mathcal{H} = \brac{0.1, 1}$ and $\pi_{H}(0.1) = 0.6$.
Each point in the tradeoff curves is obtained by numerical solution of \eqref{eq:tradeoffproblem_init_dual} via policy iteration for a model with buffer size truncated appropriately.
Each tradeoff curve is obtained by varying $\beta$.
From the asymptotic characterization of $Q^*(P_{c})$, we have that for $\lambda = 0.80$, $Q^*(P_{c})$ increases as $1/(P_{c} - 1.1071)$, while for $\lambda = 0.78$ and $0.82$, $Q^*(P_{c})$ increases as $\log\nfrac{1}{P_{c} - c(\lambda)}$.
\begin{figure}[h]
\centering
  \includegraphics[width=\journaldtImageWidth,height=\journaldtImageHeight]{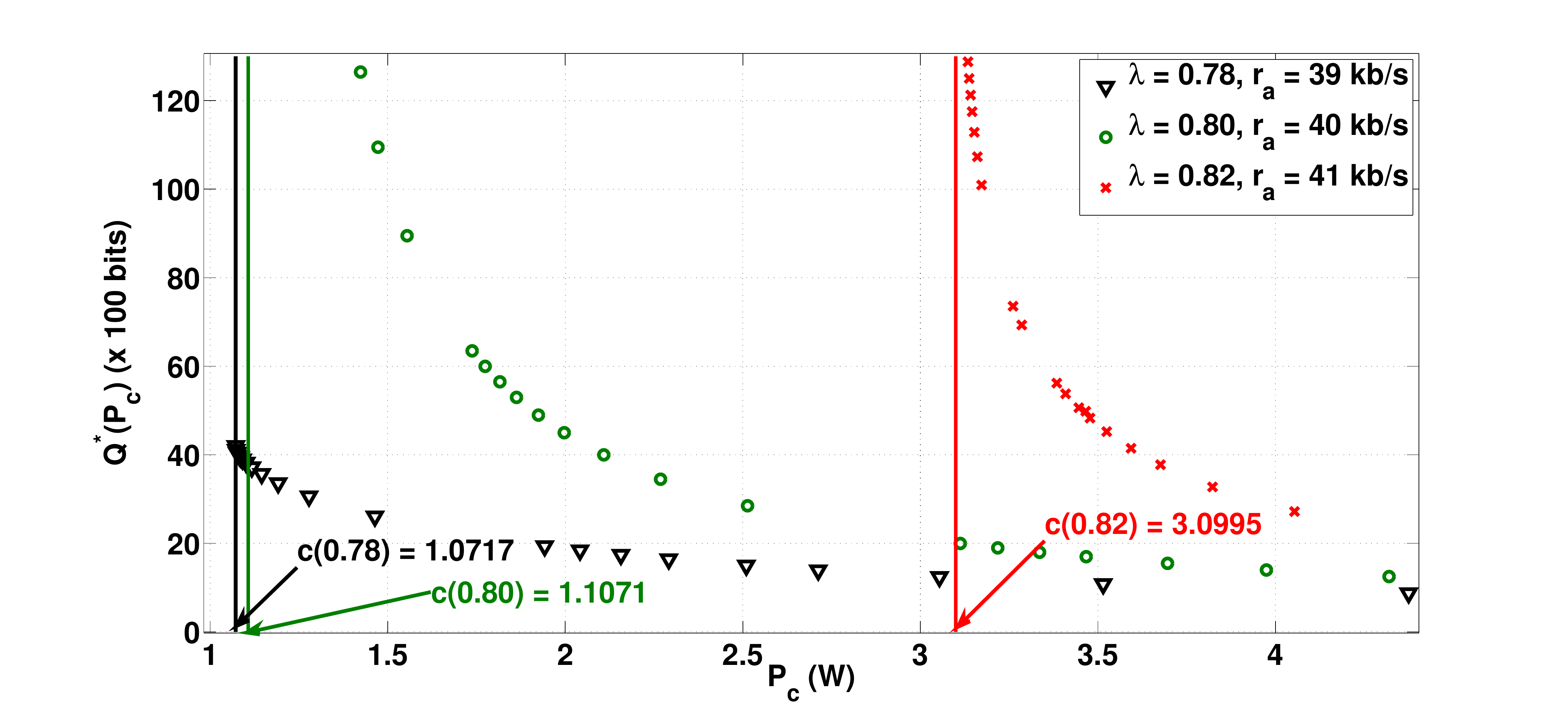}
  \caption{\small{The optimal tradeoff $Q^*(P_{c})$ for the system in Section \ref{sec:simpleeg} with $\mathcal{H} = \brac{0.1, 1}$ and $\pi_{H}(0.1) = 0.6$, for $\lambda \in \brac{0.78, 0.80, 0.82}$ with $c(0.78) = 1.0717$, $c(0.80) = 1.1071$, and $c(0.82) = 3.0995$.}}
  \label{fig:comparelambda_2}
\end{figure}

We consider the case $\lambda = 0.2$ in Figure \ref{fig:case1}, which corresponds to Case 1 for the example in Section \ref{sec:simpleeg}, for both $\mathcal{H} = \brac{0.5,1}$ and $\mathcal{H} = \brac{0.1, 1}$.
We note that $c(0.2)$ is $0.1992$ for both $\mathcal{H} = \brac{0.5,1}$ and $\mathcal{H} = \brac{0.1, 1}$.
We observe that $Q^*(P_{c})$ approaches a finite value in both cases.

\begin{figure}[h]
\centering
  \includegraphics[width=\journaldtImageWidth,height=\journaldtImageHeight]{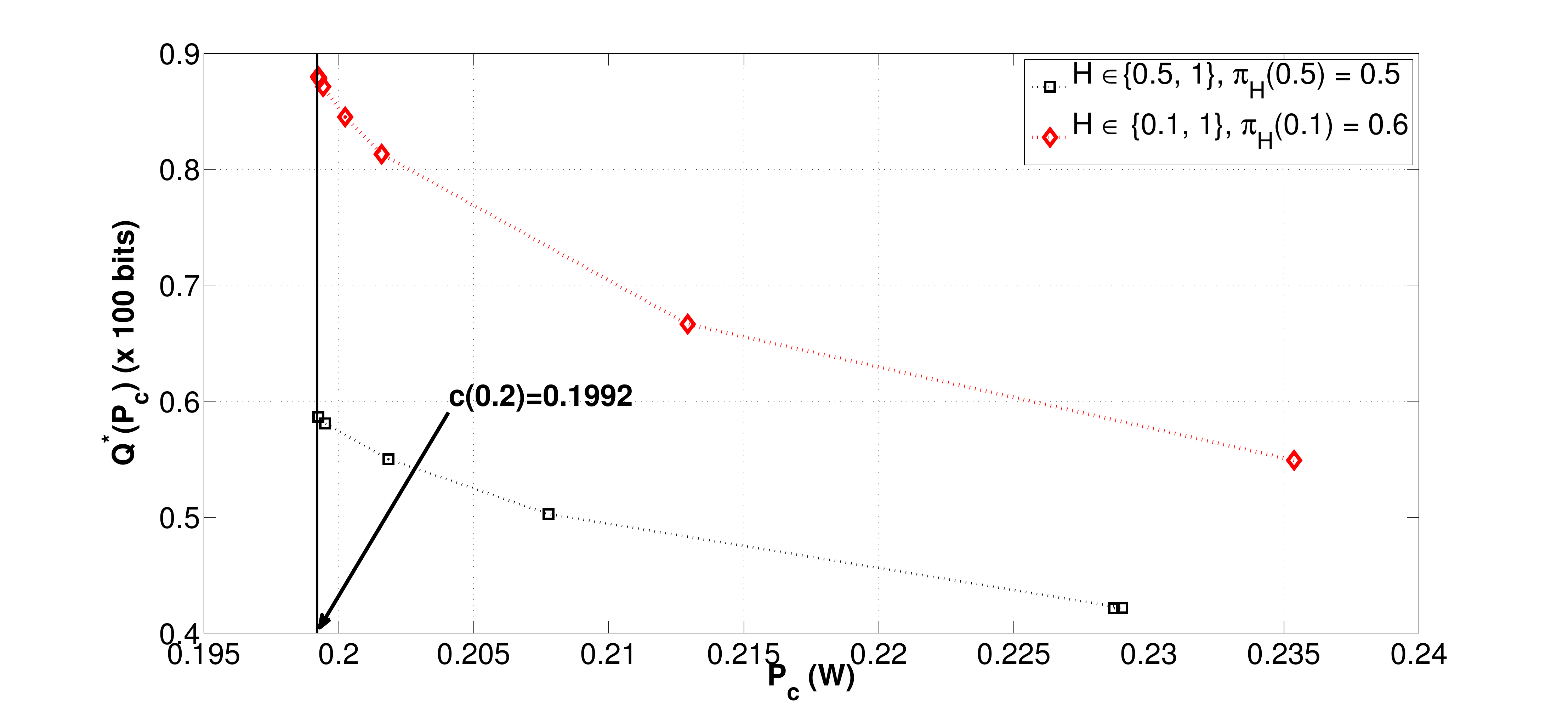}
  \caption{\small{The optimal tradeoff $Q^*(P_{c})$ for the system in Section \ref{sec:simpleeg}, for $\lambda  = 0.2$, for two cases of $\mathcal{H}$; $Q^*(P_{c})$ approaches a finite value in both cases.}}
  \label{fig:case1}
\end{figure}
\subsection{Asymptotic lower bounds for R-model}

We now consider the asymptotic behaviour of $Q^*(P_{c})$ in the asymptotic regime $\Re$ as $P_{c} \downarrow c_{R}(\lambda)$ for the R-model.
We note that the following result is similar to the Berry-Gallager lower bound, but is derived with the extra assumption G2.

\begin{proposition}
  For any sequence of admissible policies $\gamma_{k}$ with $\overline{P}(\gamma_{k}) - c_{R}(\lambda) = V_{k} \downarrow 0$, we have that $\overline{Q}(\gamma_{k}) = \Omega\nfrac{1}{\sqrt{V_{k}}}$. Therefore, $Q^*(P_{c}) = \Omega\nfrac{1}{\sqrt{P_{c} - c(\lambda)}}$ as $P_{c} \downarrow c_{R}(\lambda)$. 
\label{lemma:realvalued}
\end{proposition}
The proof is given in Appendix \ref{appendix:proof:realvalued}.
We note that an upper bound which achieves the above asymptotic lower bound, upto a logarithmic factor is obtained from a sequence of TOCA policies in \cite{neely_mac}.
Since TOCA policies are not admissible, a sequence of admissible policies that achieve the same asymptotic upper bound as TOCA is presented in \cite{vineeth_thesis}. 

\textbf{R-model with piecewise linear ${P(h,s)}$}:
\label{sec:rmodel_imodel_comparison}
We note that R-model with a strictly convex $P(h, s)$ is usually used as an approximation to I-model.
Usually, the function $P(h, s)$ for R-model coincides with $P(h, s)$ for the I-model for $s \in \brac{0, \dots, S_{max}}$, $\forall h \in \mathcal{H}$.
But then we find that there are differences in the asymptotic behaviour of $Q^*(P_{c})$ for I-model and R-model.
We observe that $c_{R}(\lambda) = c(\lambda), \forall \lambda \in \Lambda$ and $c_{R}(\lambda) < c(\lambda)$ for $\lambda \not \in \Lambda$.
R-model suggests that $Q^*(P_{c})$ increases to infinity for all $\lambda \in (0, S_{max})$ as $P_{c} \downarrow c_{R}(\lambda)$.
However, for Case 1, we see that $Q^*(c(\lambda))$ is finite (we note that $c_{R}(\lambda) < c(\lambda)$ in this case).
For Case 2, as a function of $V$, we have that the minimum average queue length increases as $\Omega\nfrac{1}{\sqrt{V}}$ for R-model, whereas for I-model the average queue length increases only as $\Omega\brap{\log\nfrac{1}{V}}$.
For Case 3, with $c_{R}(\lambda) = c(\lambda)$ we have that $Q^*(P_{c}) = \Omega\nfrac{1}{\sqrt{P_{c} - c(\lambda)}}$ for the R-model, whereas for the I-model $Q^*(P_{c}) = \Omega\nfrac{1}{{P_{c} - c(\lambda)}}$.
So R-model with a strictly convex $P(h, s)$ overestimates the behaviour of $Q^*(P_{c})$ as a function of $V$ for Cases 1 and 2, and underestimates $Q^*(P_{c})$ for Case 3.
In the following, we briefly outline a method to show that a better approximation for I-model, is R-model with a piecewise linear $P(h, s)$.
This piecewise linear transmit power function $P(h, s), s \in [0, S_{max}]$ is chosen as the lower convex envelope of the transmit power function $P(h, s), s \in \brac{0, \dots, S_{max}}$ for the I-model.
With this choice of $P(h,s)$ for R-model, it can be shown that $c_{R}(\lambda) = c(\lambda), \forall \lambda \in [0, S_{max}]$.
Then the asymptotic order behaviour of $Q^*(P_{c})$ is dependent on $\lambda$, and three cases arise.
These three cases are the same as that for I-model, defined in Section \ref{sec:differentcases}.
We now show that the asymptotic behaviour of $Q^*(P_{c})$ for I-model and R-model is the same for Cases 2 and 3.
\begin{proposition}
  For R-model, with the piecewise linear $P(h, s)$ defined as above, for any sequence $\gamma_{k} \in \Gamma_{k}$, with $\Pgk - c_{R}(\lambda) = V_{k} \downarrow 0$, we have that $\Qgk = \Omega\brap{\log\nfrac{1}{V_{k}}}$ for Case 2 and $\Qgk = \Omega\nfrac{1}{V_{k}}$ for Case 3. Therefore, $Q^*(P_{c}) = \Omega\brap{\log\nfrac{1}{P_{c} - c_{R}(\lambda)}}$ for Case 2 and $Q^*(P_{c}) = \Omega\nfrac{1}{P_{c} - c_{R}(\lambda)}$ for Case 3.
  \label{prop:rmodel_piecewisecost}
\end{proposition}
We note that the proof of this result is quite similar to that of Propositions \ref{lemma:case2} and \ref{lemma:case3lb}.
Hence, only an outline of the proof is given in Appendix \ref{appendix:proof:rmodel_piecewisecost}.
For Case 1 with a single fade state $h_{0}$, it can be shown that $\overline{Q}(\gamma_{u})$ is finite for R-model, so that $Q^*(c_{R}(\lambda))$ is also finite.
However, we do not have asymptotic lower bounds for this case.

\section{Extension: Models with admission control}
\label{sec:systemmodel_modelus}
We consider the optimal tradeoff of average queue length and average power for a fading point-to-point link, when the packets arriving to the link can be dropped, subject to a constraint on the utility of the time average throughput of the packets which are transmitted.
\subsection{System model}
We indicate only the differences from the models in Sections \ref{sysmodel:intval} and \ref{sysmodel:realval}.
We also denote the expectation with respect to the distribution of $R[1]$ as $\Exp_{R}$.
We assume that just before the end of slot $m$, $A[m] \leq R[m]$ packets are admitted into the infinite length transmitter buffer, while $R[m] - A[m]$ packets are dropped.
For this model, a policy $\gamma$ for operation of the transmitter is the sequence of service and arrival batch sizes $(S[1], A[1], S[2], A[2], \dots)$.
The set of all policies is denoted as $\Gamma$.
If $\gamma$ is such that $S[m] = S(Q[m - 1], H[m])$ and $A[m] = A(Q[m - 1], R[m], H[m])$, where $S(q, h)$ and $A(q, r, h)$ are randomized functions, then $\gamma$ is a stationary policy.
The set of stationary policies is denoted by $\Gamma_{s}$.
Since $(R[m], H[m], m \geq 1)$ is assumed to be IID, we have that for a $\gamma \in \Gamma_{s}$, $(Q[m], m \geq 0)$ is a Markov chain.

If we assume that $A[m], R[m], q_{0}, S[m] \in \sZ$, $\forall m$, then the queue evolution $Q[m] \in \sZ$ and the model is denoted as I-model-U.
On the other hand, if $A[m], R[m], q_{0}, S[m] \in \sR$, $\forall m$, then the queue evolution $Q[m] \in \sR$, and the model is denoted as R-model-U.
Similar to R-model, R-model-U with a strictly convex $P(h, s)$ function is usually used as an analytically tractable approximation for I-model-U.
We note that R-model-U has been studied in \cite{neely_utility}.

We define the average throughput of a policy $\gamma \in \Gamma_{s}$ as
\begin{eqnarray}
  \overline{A}(\gamma, q_{0}) = \liminf_{M \rightarrow \infty} \frac{1}{M} \Exp \bras{\sum_{m = 1}^{M} A[m] \middle \vert Q[0] = q_{0} }.
\end{eqnarray}
Let $u(a) : [0, A_{max}] \rightarrow \sR$ be a strictly concave and increasing function of $a$, with $u(0) = 0$.
The utility of transmitting the packets is $u(\overline{A}(\gamma, q_{0}))$, for a policy $\gamma$.
The average power for a policy $\gamma \in \Gamma_{s}$ is $\overline{P}(\gamma, q_{0})$ and the average queue length is $\overline{Q}(\gamma, q_{0})$, as defined in \eqref{chap5fading:eq:avgpower} and \eqref{chap5fading:eq:avgqlength} respectively.

\subsection{Tradeoff problem}
\label{sec:problem}

We consider the optimal tradeoff between $\overline{Q}(\gamma, q_{0})$ and $\overline{P}(\gamma, q_{0})$ subject to the average utility $u(\overline{A}(\gamma, q_{0}))$ being at least a positive $u_{c} < u(\lambda)$, for the class of stationary policies $\Gamma_{s}$, for I-model-U and R-model-U.
The constraint $u(\overline{A}(\gamma, q_{0})) \geq u_{c}$ is equivalent to having the constraint $\overline{A}(\gamma, q_{0}) \geq u^{-1}(u_{c})$, where $u^{-1}$ is the inverse function of $u$.
We denote $u^{-1}(u_{c})$ by $\rho\lambda$, where $0 < \rho < 1$.
We note that since the arrival rate is not the same for all $\gamma \in \Gamma_{s}$, minimization of the average queue length does not directly correspond to minimizing the average delay of the packets.
Asymptotic bounds on the average delay can be derived using Little's law and are discussed in the following.

The tradeoff problem that we consider is
\begin{eqnarray*}
  \mini_{\gamma \in \Gamma} & & \overline{Q}(\gamma, q_{0}) \\
  \text{ such that } & & \overline{P}(\gamma, q_{0}) \leq P_{c} \text{ and } \overline{A}(\gamma, q_{0}) \geq \rho \lambda.
  \label{eq:tradeoffutility_init}
\end{eqnarray*}
As in Section \ref{sec:cmdp_formulation} we can show that if $P_{c} > c(\rho\lambda)$ for I-model-U or if $P_{c} > c_{R}(\rho\lambda)$ for R-model-U (which are the minimum average powers required for mean rate stability while supporting an arrival rate of $\rho\lambda$ rather than $\lambda$) then there exists an optimal stationary policy $\gamma^*$.
Therefore, we can restrict ourselves to the set of stationary policies.
\subsubsection{Admissible policies}
As for I-model and R-model, we consider the above tradeoff problem for a set of admissible policies $\Gamma_{a}$.
Since there is admission control, it is not reasonable to assume that the Markov chain under a $\gamma \in \Gamma_{s}$ is irreducible (e.g., the dynamic packet dropping (DPD) policies in \cite{neely_utility} drops all packets once the queue length reaches a threshold value). 
So we relax the irreducibility requirement for I-model-U and R-model-U as follows.
For an admissible policy $\gamma$, the Markov chain $(Q[m], m \geq 0)$ has a single positive recurrent class $\mathcal{R}_{\gamma}$ which contains $0$.
Furthermore, the cumulative expected queue cost as well as the cumulative expected power cost starting from any state $q_{0}$ until $\mathcal{R}_{\gamma}$ is hit are finite.

We note that for a $\gamma \in \Gamma_{a}$, $\overline{Q}(\gamma, q_{0}) = \Qg$, $\overline{P}(\gamma, q_{0}) = \Pg$, and $\overline{A}(\gamma, q_{0}) = \Ag$.
A policy $\gamma$ is defined to be admissible if: (i) the Markov process $(Q[m], m \geq 0)$ under $\gamma$ is aperiodic and positive Harris recurrent on a single recurrence class $\mathcal{R}_{\gamma}$ with stationary distribution\footnote{We note that $\pi(A) = 0$ for any $A \not \subset \mathcal{R}_{\gamma}$} $\pi$, (ii) $\overline{Q}(\gamma, q_{0}) < \infty$, and (iii) $\sq$ is non-decreasing in $q$, where $\sq = \ExpH \Exp_{S|q,H}S(q,H)$ is the average service rate at queue length $q$.
We note that for a $\gamma \in \Gamma_{a}$,
\begin{eqnarray*}
  \Qg & = & \Expp Q, \\
  \Pg & = & \Expp \ExpH \Expsqh P(H, S(Q, H)), \\
  \Ag & = & \Expp \ExpH \Exp_{R} \Exp A(Q, R, H).
\end{eqnarray*}
Let the average service rate be $\Sg$, then $\Sg = \Expp \ExpH \Expsqh S(Q, H)$.
\subsubsection{Tradeoff problem for admissible policies}
The problem TRADEOFF that we consider is 
\begin{eqnarray*}
  \mini_{\gamma \in \Gamma_{a}} \Qg \text{ such that } \Pg \leq P_{c} \text{ and } \Ag \geq \rho \lambda.
\end{eqnarray*}
The optimal value of TRADEOFF is denoted as $Q^*(P_{c}, \rho)$.
\subsubsection{The asymptotic regime $\Re$}
Suppose $\gamma$ is feasible for TRADEOFF.
Then
\begin{eqnarray*}
  \Sg = \Ag \geq \rho\lambda.
\end{eqnarray*}
Now we note that $\Pg$ is bounded below by the optimal value of 
\begin{eqnarray}
  \mini_{\gamma \in \Gamma_{a}} & & \Expp \ExpH \Expsqh P(H, S(Q, H)), \nonumber \\
  \text{such that } & & \Expp \ExpH \Expsqh S(Q, H) \geq \rho\lambda.
  \label{eq:mincost_problem1}
\end{eqnarray}
We note that $\Expp \ExpH \Expsqh S(Q,H) = \ExpH \Expp \Expsqh S(Q,H)$.
We have that $\ExpH \Expp \Expsqh S(Q,H) = \ExpH \Exp_{S|H} S$ and $\Expp \ExpH \Expsqh P(H, S(Q,H)) = \ExpH \Exp_{S|H} P(H, S)$, where the conditional distribution of $S$ given $H$ depends on the policy $\gamma$, as in Section \ref{sec:imodelrmodel_asymp_regime}.
Then the optimal value of \eqref{eq:mincost_problem1} is bounded below by the optimal value of
\begin{eqnarray}
  \mini & & \ExpH \Exp_{S|H} P(H, S), \nonumber \\
  \text{such that } & & \ExpH \Exp_{S|H} S \geq \rho\lambda,
  \label{eq:mincost_problem2}
\end{eqnarray}
where we minimize over all possible conditional distributions for $S$ given $h \in \mathcal{H}$, irrespective of the policy $\gamma$.
We note that for I-model-U, these distributions have support on $\brac{0, \dots, S_{max}}$, whereas for R-model-U they have support on $[0, S_{max}]$.

We note that \eqref{eq:mincost_problem2} has feasible solutions only if $\rho\lambda \leq S_{max}$.
The optimal value of the above problem is $c(\rho\lambda)$ for I-model-U and $c_{R}(\rho\lambda)$ for R-model-U, since the constraint is satisfied with equality\footnote{If the distribution which achieves the minimum in \eqref{eq:mincost_problem2} is such that $\ExpH \Exp_{S|H} S > \rho\lambda$, then it is possible to show that there exists another distribution which has a strictly smaller $\ExpH \Exp_{S|H} P(H, S)$.}.
So, we have that for $\gamma \in \Gamma_{a}$, $\Pg \geq c(\rho\lambda)$ for I-model-U and $\Pg \geq c_{R}(\rho\lambda)$ for R-model-U.
Thus, TRADEOFF has feasible solutions only if $P_{c} \geq c(\rho\lambda)$ for I-model-U and $P_{c} \geq c_{R}(\rho\lambda)$ for R-model-U.

We now show that $c(\rho\lambda)$ and $c_{R}(\rho\lambda)$ are both $\inf_{\brac{\gamma : \gamma \in \Gamma_{a}, \Ag \geq \rho\lambda}} \Pg$ for I-model-U and R-model-U respectively.
For I-model-U, we consider a sequence of policies $\gamma_{k}$, where for each $\gamma_{k}$, at each slot $m$, each customer in the batch $R[m]$ is admitted with probability $\rho$ and dropped with probability $1 - \rho$.
Then $\forall \gamma_{k}$ we have that $\Agk \geq \rho\lambda$.
For $\gamma_{k}$, the batch service size $(S[m])$ is chosen (\cite[Theorem 1]{neely_mac}) as for the I-model, with arrival rate $\rho\lambda$, so that $\Pgk = c(\rho\lambda) + V_{k}$ and $\Qgk = \mathcal{O}\nfrac{1}{V_{k}}$.
Thus, we have that there exists a sequence of admissible policies $\gamma_{k}$, such that $\Pgk = c(\rho\lambda) + V_{k}$, $\Agk \geq \rho\lambda$, and $\Qgk = \mathcal{O}\nfrac{1}{V_{k}}$, for a sequence $V_{k} \downarrow 0$.
Similarly, for R-model-U, there exists a sequence of policies $\gamma_{k}$, for which we choose $A[m] = \rho R[m]$, and $(S[m])$ is as for R-model (\cite[Theorem 1]{neely_mac}, with arrival rate $\rho\lambda$), so that $\Pgk = c_{R}(\rho\lambda) + V_{k}$ and $\Qgk = \mathcal{O}\nfrac{1}{V_{k}}$.
Hence, $c(\rho\lambda)$ and $c_{R}(\rho\lambda)$ are $\inf_{\brac{\gamma : \gamma \in \Gamma_{a}, \Ag \geq \rho\lambda}} \Pg$ for I-model-U and R-model-U respectively.

In the following, we obtain an asymptotic characterization of $Q^*(P_{c}, \rho)$ in the asymptotic regimes $\Re$ as $P_{c} \downarrow c(\rho\lambda)$ for I-model-U and $P_{c} \downarrow c_{R}(\rho\lambda)$ for R-model-U, under the assumption that $\rho\lambda < S_{max}$.
We recall that $c(s)$ is a non-decreasing, piecewise linear, and convex function of $s \in [0, S_{max}]$, whereas $c_{R}(s)$ is a non-decreasing strictly convex function of $s \in [0,S_{max}]$, with $c(0)$ and $c_{R}(0)$ both being $0$.

\subsection{Asymptotic lower bound}
\label{sec:lowerbound}
We first present the intuition behind the asymptotic behaviour of $Q^*(P_{c}, \rho)$ for R-model-U in the regime where $V = P_{c} - c_{R}(\rho\lambda) \downarrow 0$.
Then we present an asymptotic lower bound for R-model-U.
We then discuss the asymptotic lower bound for I-model-U, since it can be obtained using very similar techniques as for R-model-U and as in Proposition \ref{lemma:case2}.

Intuition about the asymptotic behaviour of $Q^*(P_{c}, \rho)$ for R-model-U is obtained using the simplified M/M/1 queueing model discussed in Section \ref{sec:asymp_analysis}.
However, since for R-model-U, we have admission control, the birth rate $\lambda(q)$ is also controllable for the M/M/1 queueing model that we consider here.
Let $\mathcal{Q}_{h} = \brac{q : \mu(q) \in [\rho\lambda - \epsilon_{V}, \rho\lambda + \epsilon_{V}]}$.
Similar to R-model, it can be shown that $Pr\brac{Q \in \mathcal{Q}_{h}} \uparrow 1$ as $V \downarrow 0$.
Then we have that as $V \downarrow 0$, $\mu(q)$ for $q \in \mathcal{Q}_{h}$ approaches $\rho\lambda$.
However, we note that since $\lambda(q)$ is also controllable, instead of $\pi(q)$ being constant as in the case of R-model, $\pi(q)$ can in fact be geometrically increasing, constant, and then geometrically decreasing as in Case 2 (I-model).
Therefore, we expect that $Q^*(P_{c}, \rho)$ only grows as $\Omega\brap{\log\nfrac{1}{V}}$.

Let $\gamma$ be an admissible policy with $\Pg - c_{R}(\rho\lambda) = V$ and $\Ag \geq \rho\lambda$.
As in Section \ref{sec:gammadep_lowerbound}, we have that $\Expp c_{R}(\sQ) \leq \Pg$.
Therefore, $\Expp c_{R}(\sQ) - c_{R}(\rho\lambda) \leq \Pg - c_{R}(\rho\lambda) \leq V$.
Since $c_{R}(.)$ is a strictly convex and non-decreasing function, we assume that the second derivative of $c_{R}(s)$ is positive at $s = \rho\lambda$.
Then, we have the following result.

\begin{proposition}
  For any sequence of admissible policies $\gamma_{k}$ such that $\Agk \geq \rho\lambda$ and $\Pgk - c_{R}(\rho\lambda) = V_{k} \downarrow 0$, we have that $\Qgk = \Omega\brap{\log\nfrac{1}{V_{k}}}$. Therefore, $Q^*(P_{c}, \rho) = \Omega\brap{\log\nfrac{1}{P_{c} - c_{R}(\rho\lambda)}}$.
  \label{lemma:tradeoffutilitylb}   
\end{proposition}
The proof is given in Appendix \ref{appendix:proof:tradeoffutilitylb}.
We now discuss the asymptotic lower bound for I-model-U in the regime $\Re$.
The analysis for I-model-U proceeds in a similar fashion as in Proposition \ref{lemma:case2}; the piecewise linear function $c(s)$ and the quantities $a_{p}, p \geq 1$ are similarly defined.
The three cases which then arise are : (1) $0 < \rho \lambda < a_{2}$, (2), $a_{p} < \rho \lambda < a_{p + 1}, p > 1$, and (3) $\rho\lambda = a_{p}, p > 1$.
For Cases 2 and 3, proceeding similarly as in the proof of Proposition \ref{lemma:case2}, it is possible to show that, for any sequence of admissible policies $\gamma_{k}$ such that $\Agk \geq \rho\lambda$ and $\Pgk - c(\rho\lambda) = V_{k} \downarrow 0$, we have that $\Qgk = \Omega\brap{\log\nfrac{1}{V_{k}}}$.
We note that we do not have any asymptotic results for Case 1, although numerically it can be shown that $Q^*(P_{c}) < \infty$ even if $P_{c} = c(\rho\lambda)$.
For example, for the examples in Figure \ref{fig:case1}, we use a policy that drops packets with probability $1 - \rho$ and then uses the optimal service batch size for the cases in Figure \ref{fig:case1}.


\subsection{Discussion}
\label{sec:discussion}

\paragraph{Comparison with known results}
The model considered by Neely \cite{neely_utility} is the same as R-model-U.
It is shown in \cite{neely_utility} that there exists a sequence of policies $\gamma_{k} \in \Gamma_{s}$ with a corresponding sequence $V_{k} \downarrow 0$, such that $\overline{A}(\gamma_{k}, q_{0}) \geq \rho \lambda$, $\overline{Q}(\gamma_{k}, q_{0}) = \mathcal{O}\brap{\log\nfrac{1}{V_{k}}}$, and $\overline{P}(\gamma_{k}, q_{0})$ is at most $V_{k}$ more than the minimum average power required for queue stability.
It is also shown in \cite{neely_utility} for $|\mathcal{H}| = 1$, that if $\gamma_{k}$ is any sequence of policies, with $\overline{P}(\gamma_{k}, q_{0})$ at most $V_{k}$ more than the minimum average power required for queue stability and $\overline{A}(\gamma_{k}, q_{0}) \geq \rho \lambda$, then $\overline{Q}(\gamma_{k}, q_{0}) = \Omega\brap{\log\nfrac{1}{V_{k}}}$ as $V_{k} \downarrow 0$.
In Proposition \ref{lemma:tradeoffutilitylb}, we have derived an asymptotic lower bound for $|\mathcal{H}| > 1$, but for admissible policies.
\paragraph{Minimization of average delay}
When average delay is the performance measure under consideration, then the problem that we are interested in is
\begin{eqnarray*}
  \mini_{\gamma \in \Gamma_{a}} \frac{\Qg}{\Ag} \text{ such that } \Pg \leq P_{c} \text{ and } \Ag \geq \rho \lambda,
\end{eqnarray*}
since the average delay for $\gamma \in \Gamma_{a}$ is $\frac{\Qg}{\Ag}$ from Little's law.
Let the optimal value of the above problem be $D^*(P_{c}, \rho)$.

We note that, since $c(s)$ or $c_{R}(s)$ is a convex and non-decreasing function in $s \in [0, S_{max}]$, for any feasible admissible policy $\gamma$, we have that
\begin{eqnarray*}
  \Expp c(\sQ) \leq \Pg \leq P_{c}, \\
  c(\Expp \sQ) \leq P_{c}, \text{ or}, \\
  \Sg = \Expp \sQ \leq c^{-1}(P_{c}),
\end{eqnarray*}
where $c^{-1}$ is the inverse function of $c$ for I-model-U.
Consider any sequence $P_{c,k} \downarrow c(\rho\lambda)$ as $k \uparrow \infty$.
Since $\Ag = \Sg$, the objective function in the above optimization problem can be bounded above by $\frac{\Qg}{\rho\lambda}$ and bounded below by $\frac{\Qg}{c^{-1}(P_{c,1})}$.
A similar bound can be obtained for R-model-U.
Then, it follows that the asymptotic behaviour of $D^*(P_{c,k}, \rho)$ is the same as that of $Q^*(P_{c,k},\rho)$ as $P_{c,k} \downarrow c(\rho\lambda)$ for I-model-U and $P_{c,k} \downarrow c_{R}(\rho\lambda)$ for R-model-U.
\paragraph{Relation to the asymptotic order optimal tradeoff in \cite{neely_superfast}}
Neely \cite{neely_superfast} considers a system, with both admission control and service rate control, in which the arrival rate $\lambda$ is larger than the maximum service rate $S_{max}$.
The objective is to obtain a sequence of policies $\gamma_{k}$ which achieve an order optimal minimum average queue length $\Qgk$ as the average utility $u(\Agk)$ approaches the maximum utility value $u(S_{max})$.
We note that there is no cost associated with the service of packets in \cite{neely_superfast}.
It is shown that for any sequence of policies $\gamma_{k}$ such that $u(S_{max}) - u(\Sgk) = V_{k} \downarrow 0$, $\Qgk = \Omega\brap{\log\nfrac{1}{V_{k}}}$.
A sequence of policies $\gamma_{k}$ such that $\Qgk = \mathcal{O}\brap{\log\nfrac{1}{V_{k}}}$ and $U(\Sgk) = u(S_{max}) - V_{k}$ is also obtained.
We note that as the utility function is assumed to be strictly concave and increasing, the throughput value that maximizes the utility is $S_{max}$ itself.
For any sequence $\gamma_{k}$, if $\Ugk \uparrow u(S_{max})$ it can be shown that the probability of using a service rate less than $S_{max}$ decreases to zero.
That is, with $q_{1} \stackrel{\Delta} = \sup\brac{q : \sq \leq S_{max} - \epsilon}$, $Pr\brac{Q < q_{1}} \downarrow 0$.
Therefore, the proof of Proposition \ref{lemma:tradeoffutilitylb} can be applied to obtain an alternate proof for the asymptotic logarithmic lower bound on the average queue length obtained in \cite{neely_superfast}, but for admissible policies.

\section{Summary and conclusions}
\label{sec:conclusions}
In this paper, we obtained asymptotic lower bounds for some cases, for which lower bounds were not previously available.
Our results are summarized in Table \ref{table:ourlowerbounds}.
For case 3 in Table \ref{table:asymptotic_bounds}, we obtained an $\Omega\brap{\log\nfrac{1}{V}}$ asymptotic lower bound.
We obtained an $\Omega\nfrac{1}{V}$ lower bound for case 4 in Table \ref{table:asymptotic_bounds} which was previously not available.
We obtained an $\Omega\brap{\log\nfrac{1}{V}}$ asymptotic lower bound for case 5 in Table \ref{table:asymptotic_bounds}, which holds for $|\mathcal{H}| > 1$.
We note that these bounds were derived for admissible policies.
In \cite[Lemmas 4.3.23, 5.6.1]{vineeth_thesis}, we have also obtained asymptotic lower bounds for $Q^*(P_{c})$ for I-model, when the arrival process and fading process are ergodic rather than IID.
\begin{table*}
  \centering
  \begin{tabular}{|l|l|l|}
    \hline
    \textbf{Model details} &
    \begin{minipage}{0.7\textwidth}
      \textbf{Results (in the regime $\Re$, for admissible policies)}
    \end{minipage} \\
    \hline
    \begin{minipage}{0.1\textwidth}
      \vspace{0.2em}
        I-model
      \vspace{0.2em}
    \end{minipage} & 
    \begin{minipage}{0.7\textwidth}
      Depending on the arrival rate $\lambda$, minimum average queue length either increases to only a finite value, or is $\Omega\brap{\log\nfrac{1}{V}}$ or is $\Omega\nfrac{1}{{V}}$
    \end{minipage} \\
    \hline
    \begin{minipage}{0.1\textwidth}
      \vspace{0.2em}
      R-model
      \vspace{0.2em}
    \end{minipage} & 
    \begin{minipage}{0.7\textwidth}
      Minimum average queue length is $\Omega\nfrac{1}{\sqrt{V}}$ (previously known \cite{berry} but re-derived here using our method).
    \end{minipage} \\
    \hline
    \begin{minipage}{0.1\textwidth}
      \vspace{0.2em}
      I-model-U
      \vspace{0.2em}
    \end{minipage} & 
    \begin{minipage}{0.7\textwidth}
      Depending on the arrival rate $\lambda$, minimum average queue length either increases to only a finite value, or  is $\Omega\brap{\log\nfrac{1}{{V}}}$
    \end{minipage} \\
    \hline
    \begin{minipage}{0.1\textwidth}
      \vspace{0.2em}
      \mbox{R-model-U}
      \vspace{0.2em}
    \end{minipage} & 
    \begin{minipage}{0.7\textwidth}
      Minimum average queue length is $\Omega\brap{\log\nfrac{1}{{V}}}$
    \end{minipage} \\
    \hline
  \end{tabular}
  \caption{The asymptotic results derived in this paper for discrete time queueing models with fading}
  \label{table:ourlowerbounds}
\end{table*}

We have also presented an intuitive explanation for the behaviour of $Q^*(P_{c})$ in the asymptotic regime $\Re$, using the behaviour of the stationary probability distribution of the queue length for admissible policies.
The intuition for the behaviour of $Q^*(P_{c})$ is obtained by the analysis of a simpler state dependent M/M/1 model in \cite[Chapters 2 and 3]{vineeth_thesis}.
The intuitive approach that we have followed for obtaining lower bounds to $Q^*(P_{c})$, is used in \cite{vineeth_thesis} to obtain asymptotic bounds for the minimum average queue length for other communication scenarios, such as: (i) for a multiple access model (similar to that in \cite{neely_mac}), where we observe that average queue lengths for different transmitters may grow at different rates, (ii) for a model with just admission control, and (iii) for a model with general holding costs.
In \cite[Chapter 6]{vineeth_thesis}, we apply the analysis in this paper to obtain an asymptotic characterization of the tradeoff between average delay and average error probability for a noisy point-to-point link, with a single fade state $h_{0}$.
We assume that the link uses random block codes of length $N_{c}$ to encode and transmit bits which arrive into the transmitter queue.
We consider a particular cost function $P(h_{0}, s)$, which is an approximation of the expected number of bits which are decoded in error when $s$ bits are encoded and transmitted using a random block code with block length $N_{c}$.
The approximation is obtained using Gallager's random coding upper bound \cite{gallager_text}.
The function $P(h_{0}, s)$ is observed to be a non-convex function in $s$, $s \in \brac{0, \dots, S_{max}}$.
However, the minimum average error probability incurred with the queue being mean rate stable, $c(\lambda)$, is a piecewise linear convex function of $\lambda$, as for the I-model.
Then the asymptotic lower bounds on the minimum average delay, derived for I-model, apply as the average error probability approaches $c(\lambda)$.
The approach using bounds on the stationary probability distribution, also has the added advantage of providing asymptotic bounds on the structure of any stationary deterministic policy, as shown in \cite[Propositions 2.3.16 and 2.3.17, Lemma 3.2.18, Lemma 4.3.21]{vineeth_thesis}.
These asymptotic bounds are independent of the exact nature of the service cost function.

We recall that R-model with the strictly convex $P(h,s)$ function, is usually used as an approximation for I-model.
We observe that the behaviour of $Q^*(P_{c})$ in the regime $\Re$ is different for the approximate and original models.
We find that R-model underestimates the behaviour of $Q^*(P_{c})$ for cases 2 and 3.
We conclude that a more appropriate real valued approximate queueing model, would be one for which the service cost function $P(h,s)$ is the piecewise linear lower convex envelope of the service cost function for the original integer valued queueing model, rather than the strictly convex function assumed for R-model.

\bibliographystyle{plain}
\bibliography{journaldt}

\begin{thebibliography}{10}

\bibitem{agarwal}
M.~Agarwal, V.S. Borkar, and A.~Karandikar.
\newblock Structural properties of optimal transmission policies over a
  randomly varying channel.
\newblock {\em IEEE Transactions on Automatic Control}, 53(6), 2008.

\bibitem{altman}
E.~Altman.
\newblock {\em Constrained Markov decision processes}.
\newblock Chapman and Hall, 1999.

\bibitem{ata_pcstatic}
B.~Ata.
\newblock Dynamic power control in a wireless static channel subject to a
  quality-of-service constraint.
\newblock {\em Operations Research}, 53(5), 2005.

\bibitem{atamm1}
B.~Ata and S.~Shneorson.
\newblock Dynamic control of a {M/M/1} service system with adjustable arrival
  and service rates.
\newblock {\em Management Science}, 52(11), 2006.

\bibitem{berry_thesis}
R.A. Berry.
\newblock {\em Power and delay tradeoffs in fading channels}.
\newblock PhD thesis, LIDS, Massachusetts Institute of Technology, 2000.

\bibitem{berry}
R.A. Berry and R.G. Gallager.
\newblock Communication over fading channels with delay constraints.
\newblock {\em IEEE Transactions on Information Theory}, 48(5), May 2002.

\bibitem{gamarnik_1}
D.~Bertsimas, D.~Gamarnik, and J.N. Tsitsiklis.
\newblock Performance of multiclass {Markovian} queueing networks via piecewise
  linear {Lyapunov} functions.
\newblock {\em The Annals of Applied Probability}, 11(4), 2001.

\bibitem{bettesh}
I.~Bettesh and S.~Shamai.
\newblock Optimal power and rate control for minimal average delay: the
  single-user case.
\newblock {\em IEEE Transactions on Information Theory}, 52(9), sep. 2006.

\bibitem{chaporkar_alex}
P.~Chaporkar and A.~Proutiere.
\newblock Adaptive network coding and scheduling for maximizing throughput in
  wireless networks.
\newblock In {\em Mobicom}, 2007.

\bibitem{chen}
Wei Chen, Dayu Huang, A.A. Kulkarni, J.~Unnikrishnan, Quanyan Zhu, P.~Mehta,
  S.~P. Meyn, and A.~Wierman.
\newblock Approximate dynamic programming using fluid and diffusion
  approximations with applications to power management.
\newblock In {\em Proceedings of the 48th IEEE Conference on Decision and
  Control}, 2009.

\bibitem{collins}
B.~Collins and R.~Cruz.
\newblock Transmission policies for time varying channels with average delay
  constraints.
\newblock In {\em Proceedings of 1999 Allerton Conf. Communication, Control,
  and Computing}.

\bibitem{denteneer}
D.~Denteneer, A.J.E.M. Janssen, and J.S.H. van Leeuwaarden.
\newblock Moment inequalities for the discrete-time bulk service queue.
\newblock {\em Mathematical Methods of Operations Research}, 61, 2005.

\bibitem{gallager_text}
R.~G. Gallager.
\newblock {\em Information Theory and Reliable Communication}.
\newblock Wiley, 1968.

\bibitem{george}
J.~M. George and J.~M. Harrison.
\newblock Dynamic control of a queue with adjustable service rate.
\newblock {\em Operations Research}, 49(5), 2001.

\bibitem{hernandez_2}
J.~Gonzalez-Hernandez and C.~E. Villarreal.
\newblock Optimal policies for constrained average-cost {Markov} decision
  processes in {Borel} spaces.
\newblock {\em SIAM Journal on Control and Optimization}, 42(2), May 2003.

\bibitem{hernandez}
J.~Gonzalez-Hernandez and C.~E. Villarreal.
\newblock Optimal policies for constrained average-cost {Markov} decision
  processes.
\newblock {\em TOP Journal of Spanish Society of Statistics and Operations
  Research}, 19(1), July 2011.

\bibitem{munish}
M.~Goyal, A.~Kumar, and V.~Sharma.
\newblock Optimal cross-layer scheduling of transmissions over a fading
  multiaccess channel.
\newblock {\em IEEE Transactions on Information Theory}, 54(8), 2008.

\bibitem{huang_neely}
L.~{Huang} and M.J. {Neely}.
\newblock {Max-Weight Achieves the Exact \$[O(1/V), O(V)]\$ Utility-Delay
  Tradeoff Under Markov Dynamics}.
\newblock {\em ArXiv e-prints, available online at
  http://arxiv.org/abs/1008.0200/}, August 2010.

\bibitem{ma}
D.J. Ma, A.~M. Makowski, and A.~Shwartz.
\newblock Estimation and optimal control for constrained {Markov} chains.
\newblock {\em IEEE Conference on Decision and Control}, 1986.

\bibitem{neely_superfast}
M.J. Neely.
\newblock Super-fast delay tradeoffs for utility optimal fair scheduling in
  wireless networks.
\newblock {\em IEEE Journal on Selected Areas in Communications}, 24(8), 2006.

\bibitem{neely_mac}
M.J. Neely.
\newblock Optimal energy and delay tradeoffs for multiuser wireless downlinks.
\newblock {\em IEEE Transactions on Information Theory}, 53(9), Sept. 2007.

\bibitem{neely_utility}
M.J. Neely.
\newblock Intelligent packet dropping for optimal energy-delay tradeoffs in
  wireless downlinks.
\newblock {\em IEEE Transactions on Automatic Control}, 54(3), March 2009.

\bibitem{neely}
M.J. Neely.
\newblock {\em Stochastic network optimization with application to
  communication and queueing systems}.
\newblock Morgan and Claypool, 2010.

\bibitem{venkialtman}
V.~Ramaiyan, E.~Altman, and A.~Kumar.
\newblock Delay optimal scheduling in a two-hop vehicular relay network.
\newblock {\em Mobile Networks and Applications}, 15(1), 2010.

\bibitem{vineeth_thesis}
Vineeth~B. S.
\newblock {\em On the tradeoff of average delay, average service cost, and
  average utility for single server queues with monotone policies}.
\newblock PhD thesis, Dept. of Electrical Communication Engg., Indian Institute
  of Science, 2013.

\bibitem{vineeth_ncc13_dt}
Vineeth~B. S. and U.~Mukherji.
\newblock Tradeoff of average power and average delay for a point-to-point link
  with fading.
\newblock In {\em Proceedings of the National Conference on Communications
  (NCC), New Delhi, India}, 2013.

\bibitem{sennott}
L.~I. Sennott.
\newblock Constrained average cost {Markov} decision chains.
\newblock {\em Probability in the Engineering and Informational Sciences}, 7,
  1993.

\bibitem{weber}
S.~Stidham and R.R. Weber.
\newblock Monotonic and insensitive optimal policies for the control of queues
  with undiscounted costs.
\newblock {\em Operations Research}, 37, 1989.

\bibitem{botan}
Bo~Tan and R.~Srikant.
\newblock Online advertisement, optimization and stochastic networks.
\newblock {\em CoRR}, abs/1009.0870, 2010.

\bibitem{elif}
E.~Uysal-Biyikoglu, B.~Prabhakar, and A.~El~Gamal.
\newblock Energy-efficient packet transmission over a wireless link.
\newblock {\em IEEE/ACM Transactions on Networking}, 10(4), 2002.

\end{thebibliography}

\begin{appendices}

  \section{Proofs for Section \ref{sec:upperbounds}}
  \begin{lemma}
  Let $X$ be a random variable, with $\Exp X \geq \overline{x}$ and $X \in [0, x_{max}]$.
  Consider any $\Delta$ such that $0 < \Delta < \overline{x}$.
  Then 
  \[Pr\brac{X \geq \Delta} \geq \frac{\overline{x} - \Delta}{x_{max} - \Delta}.\]
  \label{lemma:serviceprob_bound_general}
\end{lemma}

\begin{proof}
  Let $P(x)$ be the CDF of $X$.
  Then we have that
  \begin{eqnarray*}
    \overline{x} & \leq & \int_{0}^{x_{max}} x dP(x), \\
    & \leq & \int_{0}^{\Delta^{-}} \Delta dP(x) + \int_{\Delta}^{x_{max}} x_{max} dP(x), \text{ or},\\
    & = & \Delta\brap{1 - Pr\brac{X \geq \Delta}} + x_{max} Pr\brac{X \geq \Delta}.
  \end{eqnarray*}
  Therefore, we have that
  \[Pr\brac{X \geq \Delta} \geq \frac{\overline{x} - \Delta}{x_{max} - \Delta}.\]
\end{proof}

We note that in addition, if $X$ is a discrete random variable taking values in $\sZ$, then proceeding as above, we have that 
\begin{eqnarray}
  Pr\brac{X \geq 1} \geq \frac{\overline{x}}{x_{max}}.
  \label{eq:serviceprob_bound_discrete}
\end{eqnarray}

\subsection{Proof of Proposition \ref{prelimresults:txprob_real}}
\label{app:prelimresults:txprob_real}
  We derive the upper bound for R-model first.
  As $\gamma$ is admissible we have that
  \[ Pr\brac{Q < q_{1}} = \int_{0}^{\infty} P(q, [0,q_{1})) d\pi(q), \]
  where $P(q,\mathcal{Q})$ is the transition kernel of the Markov chain.
  Hence, we have that
  \small
  \begin{eqnarray*}
    Pr\brac{Q < q_{1}} & \geq & \int_{q_{1}}^{\brap{q_{1} + \Delta}^{-}} P(q, [0, q_{1})) d\pi(q),
  \end{eqnarray*}
  \begin{eqnarray*}
    & \geq & Pr\brac{A[1] = 0} Pr\brac{S(q) \geq \Delta} Pr\brac{q_{1} \leq Q < q_{1} + \Delta}, \\
    & \geq & Pr\brac{A[1] = 0} \delta_{s} Pr\brac{q_{1} \leq Q < q_{1} + \Delta},
  \end{eqnarray*}
  \normalsize
  where $S(q) = S(q, H)$ and $\delta_{s} = \frac{s_{1} - \Delta}{S_{max} - \Delta}$ from Lemma \ref{lemma:serviceprob_bound_general}.
  Let $\rho_{d} = Pr\brac{A[1] = 0} \delta_{s}$.
  Also, for any $q' > q$, let us denote $Pr\brac{q \leq Q < q'}$ by $\pi[q, q')$.
  
  Then we have obtained that
  \[ \pi[0, q_{1}) \geq \rho_{d} \pi\left[q_{1}, q_{1} + \Delta\right). \]
  Similarly, we have that
  \[ \pi\left[0, q_{1} + \Delta\right) \geq \rho_{d} \pi\left[q_{1} + \Delta, q_{1} + 2\Delta\right), \]
  which can be written as
  \begin{eqnarray*}
    \pi[0, q_{1}) + \pi\left[q_{1}, q_{1} + \Delta\right) & \geq & \rho_{d} \pi\left[q_{1} + \Delta, q_{1} + 2\Delta\right), \\
    \pi[0, q_{1})\bras{1 + \frac{1}{\rho_{d}}} & \geq & \rho_{d} \pi\left[q_{1} + \Delta, q_{1} + 2\Delta\right).
  \end{eqnarray*}
  By induction, for $m \geq 0$, we have that
  \begin{eqnarray*}
    \frac{\pi[0, q_{1})}{\rho_{d}} \brap{1 + \frac{1}{\rho_{d}}}^{m} & \geq & \pi\left[q_{1} + m \Delta, q_{1} + (m + 1) \Delta\right).
  \end{eqnarray*}
  To obtain the upper bound for I-model, we proceed as above, with $\Delta = 1$.
  However, we use the bound \eqref{eq:serviceprob_bound_discrete} on $Pr\brac{S(q) \geq 1}$.
  Then we have that $\delta_{s} = \frac{s_{1}}{S_{max}}$.
  We then obtain that for any $k \geq 0$ and $q = q_{1} + k$,
  \begin{equation*}
    \pi(q) \leq Pr\brac{Q < q_{1}} \frac{\left(1 + \frac{1}{\rho_{d}}\right)^{k}}{\rho_{d}},
  \end{equation*}
  where $\rho_{d} \stackrel{\Delta} = \nfrac{s_{1}}{S_{max}}Pr\brac{A[1] = 0}$.\qeda

\subsection{Proof of Proposition \ref{prelimresults:drift_real}}
\label{app:prelimresults:drift_real}
  We derive the lower bound for R-model first.
  We define $\widehat{Q}[m] = \max(q_{1} + \Delta, Q[m])$ and $\widehat{Q} = \max(q_{1} + \Delta, Q)$, where $\Delta > 0$ will be chosen in the following.
  We note that since the policy is admissible and $q_{1}$ is finite, $\Expp \widehat{Q} < \infty$.
  Therefore
  \begin{equation*}
    \int_{0}^{\infty} \Exp \bras{\widehat{Q}[m + 1] - \widehat{Q}[m] \middle \vert Q[m] = q} d\pi(q) = 0.
  \end{equation*}
  As in \cite{gamarnik_1}, we split the above integral into three parts which leads to 
  \begin{eqnarray}
    & 0 = \int_{0}^{q_{1}^{-}} \mathbb{E}\bras{\widehat{Q}[m + 1] - \widehat{Q}[m] \middle | Q[m] = q} d\pi(q) 
    \label{chap5:eq:real0} \\
    & + \int_{q_{1}}^{\brap{q_{1} + \Delta}^{-}} \mathbb{E}\bras{\widehat{Q}[m + 1] - \widehat{Q}[m] \middle | Q[m] = q} d\pi(q) 
    \label{chap5:eq:real1} \\
    & + \int_{q_{1} + \Delta}^{\infty} \mathbb{E}\bras{\widehat{Q}[m + 1] - \widehat{Q}[m] \middle | Q[m] = q } d\pi(q).
    \label{chap5:eq:real2}
  \end{eqnarray}
  We note that for $q \in [0, q_{1})$ we have that $\widehat{Q}[m] = q_{1} + \Delta$ and $\widehat{Q}[m + 1] \geq q_{1} + \Delta$, so that \eqref{chap5:eq:real0} $\geq 0$.
  Suppose $q_{1} + \Delta \leq q_{d}$. 
  Since $\widehat{Q}[m] \geq Q[m]$ we obtain that \eqref{chap5:eq:real2}
  \begin{eqnarray*}
    & \geq & - d Pr\brac{q_{1} + \Delta \leq Q \leq q_{d}} + \\
    & & \int_{q_{d}^{+}}^{\infty} \mathbb{E}\bras{{Q}[m + 1] - {Q}[m] | Q[m] = q} d\pi(q)
  \end{eqnarray*}
  To obtain a lower bound on \eqref{chap5:eq:real1} we note that for $q < q_{1} + \Delta$, $\widehat{Q}[m] = q_{1} + \Delta$ and $\widehat{Q}[m + 1] \geq q_{1} + \Delta$.
  So $\widehat{Q}[m + 1] - \widehat{Q}[m] \geq 0$.
  Then as in \cite[steps (34), (35), and (36)]{berry} we use Markov inequality to lower bound $\mathbb{E}\bras{\widehat{Q}[m + 1] - \widehat{Q}[m] \middle | Q[m] = q} , q \in [q_{1}, q_{1} + \Delta)$.
  \begin{eqnarray*}
    & & \mathbb{E}\bras{\widehat{Q}[{m + 1}] - \widehat{Q}[{m}] \middle | Q[{m}] = q} \geq \\
    & & \delta Pr\brac{\widehat{Q}[{m + 1}] - \widehat{Q}[{m}] \geq \delta \middle |Q[{m}] = q} \geq \\
    & & \delta Pr\brac{{Q}[{m + 1}] - {Q}[{m}] \geq \delta + \Delta \middle | Q[m] = q} = \\
    & & \delta Pr\brac{A[m + 1] - S[m + 1] \geq \delta + \Delta \middle \vert Q[m] = q} \geq \\
    & & \delta Pr\brac{A[m + 1] - S_{max} \geq \delta + \Delta \middle \vert Q[m] = q} \geq \delta \epsilon_{a}.
  \end{eqnarray*}
  We note that $\Delta$ and $\delta$ have to be chosen so that $\Delta + \delta < \delta_{a}$.
  Thus we obtain that \eqref{chap5:eq:real1} $\geq \delta \epsilon_{a} Pr\brac{q_{1} \leq Q < q_{1} + \Delta}$.
  Combining these bounds and using $\mathbb{E}\bras{{Q}[m + 1] - {Q}[m] | Q[m] = q} = \lambda - \sq$, we obtain that 
  \begin{eqnarray*}
    0 \geq \delta \epsilon_{a} Pr\brac{q_{1} \leq Q < q_{1} + \Delta} - d Pr\{q_{1} + \Delta \leq Q \leq q_{d}) + \\
    \int_{q_{d}^{+}}^{\infty} \mathbb{E}\bras{{Q}[m + 1] - {Q}[m] | Q[m] = q} d\pi(q), \\
    Pr\brac{Q \geq q_{1} + \Delta} \geq \depn Pr\brac{Q \geq q_{1}} + \\
    \frac{1}{\delta\epsilon_{a} + d}\bras{d Pr\brac{Q > q_{d}} - \int_{q_{d}^{+}}^{\infty} (\sq - \lambda)d\pi(q)}
  \end{eqnarray*}

  By induction, we obtain that if $k \geq 1$, and $q_{1} + k\Delta \leq q_{d}$, then $Pr\brac{Q \geq q_{1} + k\Delta}$
  \begin{eqnarray*}
    & \geq \brap{\depn}^{k} Pr\brac{Q \geq q_{1}} + \\
    & \frac{1 - \brap{\depn}^{k}}{d} \bras{d Pr\brac{Q > q_{d}} - \int_{q_{d}^{+}}^{\infty} (\sq - \lambda)d\pi(q)}, \\
    & = \brap{\depn}^{k} Pr\brac{Q \geq q_{1}} + \\
    & \brap{1 - \brap{\depn}^{k}} \bras{Pr\brac{Q > q_{d}} - \frac{1}{d}\int_{q_{d}^{+}}^{\infty} (\sq - \lambda)d\pi(q)}
  \end{eqnarray*}
  We note that the bound for $k = 0$ holds trivially.
  For I-model, we proceed similarly, except that we choose $\Delta = 1$, and \eqref{chap5:eq:real1} is bounded below by $\epsilon_{a}\pi(q_{1})$.
\qeda

\subsection{Proof of Lemma \ref{lemma:prelimresults:serviceratebound}}
\label{app:prelimresults:serviceratebound}

We consider the I-model first.
We note that for $\mathcal{S}$ as in the Lemma, $Pr\brac{\sQ \in \mathcal{S}} \leq Pr\brac{\sQ \in \overline{\mathcal{S}}}$, where $\overline{\mathcal{S}} = [0, s_{l} - \epsilon_{V}) \bigcup (s_{u} + \epsilon_{V}, S_{max}]$ for positive $\epsilon_{V}$.
From the definition of $l(s)$ we have that $\mathbb{E}_{\pi} \Exp \bras{c(S(Q, H)) - l(S(Q, H))} \leq V$.
From the convexity of $c(s)$ and the linearity of $l(s)$ we have
\begin{eqnarray*}
  \sum_{q = 0}^{\infty} \pi(q) \bras{ c(\sq) - l(\sq) } \leq V.
\end{eqnarray*}
Let $Q_{s_{l}} = \brac{q : \sq \in [0, s_{l} - \epsilon_{V})}$ and $Q_{s_{u}} = \brac{q : \sq \in (s_{u} + \epsilon_{V}, S_{max}]}$.
Then we have that\footnote{If $s_{l} = 0$ or $s_{u} = S_{max}$, then the corresponding sum in the following is zero.}
\begin{eqnarray*}
  & & \sum_{q \in Q_{s_{l}}} \pi(q) \bras{ c(\sq) - l(\sq) } + \\ 
  & & \sum_{q \in Q_{s_{u}}} \pi(q) \bras{ c(\sq) - l(\sq) } \leq V.
\end{eqnarray*}
Let $a_{p_{l}} = s_{l}$ and $a_{p_{u}} = s_{u}$.
We note that for $q \in Q_{s_{l}}$, $c(\sq) - l(\sq) \geq m_{l}\brap{s_{l}  - \sq}$, where $m_{l}$ is the tangent of the angle made by the line through $\brap{a_{p_{l} - 1}, c(a_{p_{l} - 1})}$ and $\brap{s_{l}, c(s_{l})}$ with $l(s)$.
Similarly, for $q \in Q_{s_{u}}$, $c(\sq) - l(\sq) \geq m_{u} \brap{\sq - s_{u}}$, where $m_{u}$ is the tangent of the angle made by the line through $\brap{a_{p_{u} + 1}, c(a_{p_{u} + 1})}$ and $\brap{s_{u}, c(s_{u})}$ with $l(s)$.
Then, we have that
\begin{eqnarray*}
  m_{l} \sum_{q \in Q_{s_{l}}} \pi(q) \bras{ s_{l} - \sq } + m_{u} \sum_{q \in Q_{s_{u}}} \pi(q) \bras{ \sq - s_{u} } \leq V.
\end{eqnarray*}
Let $m = \min(m_{l}, m_{u})$.
Then by definition of the sets $Q_{s_{l}}$ and $Q_{s_{u}}$, we have that
\begin{eqnarray*}
  m \epsilon_{V} \bras{\sum_{q \in Q_{s_{l}}} \pi(q) + \sum_{q \in Q_{s_{u}}} \pi(q)} \leq V, \text{ or}, \\
  Pr\brac{Q \in \mathcal{Q}_{\overline{\mathcal{S}}}} \leq \frac{V}{m \epsilon_{V}}.
\end{eqnarray*}
Therefore, $\Pr\brac{Q \in \mathcal{Q}_{\mathcal{S}}} \leq \frac{V}{m \epsilon_{V}}$.
We proceed similarly for R-model.
We note that $\Exp_{\pi} \bras{c_{R}(\sQ) - l(\sQ)} \leq V$.
Then, as in \cite[step (41)]{berry}, we have that $c_{R}(\sq) - l(\sq) = G(\sq - \lambda)$, where $G(x)$ is a strictly convex function such that $G(0) = 0, G'(0) = 0$ and $G''(0) > 0$.
Therefore, there exists a positive $a$ such that $ax^{2} \leq G(x)$.
Then, we have that 
\begin{eqnarray*}
  a\int_{0}^{\infty} \brap{\sq - \lambda}^{2} d\pi(q) \leq V, \text{ or}, \\
  \int_{q \in \mathcal{Q}_{\overline{\mathcal{S}}}} \brap{\sq - \lambda}^{2} d\pi(q) \leq \frac{V}{a}.
\end{eqnarray*}
Since for $q \in \mathcal{Q}_{\overline{\mathcal{S}}}$, $|\sq - \lambda| > \epsilon_{V}$, we have that $\Pr\brac{Q \in \mathcal{Q}_{\overline{\mathcal{S}}}} \leq \frac{V}{a\epsilon_{V}^{2}}$. Therefore, $Pr\brac{Q \in Q_{\mathcal{S}}} \leq \frac{V}{a\epsilon_{V}^{2}}$. \qeda

  \section{Proofs for Sections \ref{sec:asymp_analysis} and \ref{sec:systemmodel_modelus}}
  \subsection{Proof of Proposition \ref{lemma:case2}}
\label{appendix:proof:case2lb}

We consider a particular policy $\gamma$ in the sequence $\gamma_{k}$ with $V_{k} = V$.
Let $q_{1} \Deq \inf\brac{q : \sq \geq s_{l} - \epsilon_{1}}$ for a positive $\epsilon_{1} < s_{l}$.
From Proposition \ref{prelimresults:txprob_real}, we have that for $q = q_{1} + k, k \geq 0$
\begin{equation*}
  \pi(q) \leq Pr\brac{Q < q_{1}} \frac{\left(1 + \frac{1}{\rho_{d}}\right)^{k}}{\rho_{d}} = Pr\brac{Q < q_{1}} \frac{\rho^{k}}{\rho_{d}},
\end{equation*}
where $\rho_{d} = \nfrac{s_{l} - \epsilon_{1}}{S_{max}}Pr\brac{A[1] = 0}$ and $\rho = 1 + \frac{1}{\rho_{d}}$.

We now obtain a lower bound $\frac{\bar{q}}{2}$ on the average queue length for a policy $\gamma$, where $\bar{q}$ is such that $Pr\brac{Q \leq \bar{q}} \leq \frac{1}{2}$.
In fact, let $\bar{q} = \sup\{q : \sum_{q' = 0}^{q} \pi(q') \leq \frac{1}{2}\}$.
From Lemma \ref{lemma:prelimresults:serviceratebound}, we have that $Pr\brac{Q < q_{1}} \leq \frac{V}{m\epsilon_{1}}$.
Using the above bound on $Pr\brac{Q < q_{1}}$, we have that for sufficiently small $V$, $Pr\brac{Q < q_{1}}\brap{1 + \frac{\rho}{\rho_{d}}} < \frac{1}{2}$.
Let $\bar{q}_{1}$ be the largest integer such that
\begin{equation}
  Pr\brac{Q < q_{1}} + Pr\brac{Q < q_{1}}\sum_{q = q_{1}}^{\bar{q}_{1}} \frac{\rho^{q - q_{1}}}{\rho_{d}} \leq \frac{1}{2}.
  \label{chap5:eq:barpiq1}
\end{equation}
Then $\bar{q}_{1} \leq \bar{q}$.
We note that \eqref{chap5:eq:barpiq1} is equivalent to finding the largest $\bar{q}_{1}$ such that
\begin{eqnarray*}
  Pr\brac{Q < q_{1}}\left[ 1 + \frac{1}{\rho_{d}} \frac{\rho^{\bar{q}_{1} - q_{1} + 1} - 1}{\rho - 1} \right] \leq \frac{1}{2}, \\
  \text{or } \bar{q}_{1} \leq \log_{\rho} \left[ 1 + \rho_{d}\brap{\rho - 1} \left( \frac{1}{2Pr\brac{Q < q_{1}}} - 1 \right)\right].
\end{eqnarray*}
Hence we obtain that the $\bar{q}_{1}$ is at least
\begin{equation*}
  \log_{\rho} \left[ \frac{1}{2Pr\brac{Q < q_{1}}}\right] - 1.
\end{equation*}
Since $\overline{Q}(\gamma) \geq \frac{\bar{q}}{2} \geq \frac{\bar{q}_{1}}{2}$, we have that
\begin{equation*}
  \overline{Q}(\gamma) \geq \frac{1}{2} \left[ \log_{\rho} \left[ \frac{1}{2Pr\brac{Q < q_{1}}}\right] - 1 \right].
\end{equation*}
Since $Pr\brac{Q < q_{1}} \leq \frac{V}{m\epsilon_{1}}$,  we have that $\overline{Q}(\gamma_{k}) = \Omega\left(\log\left(\frac{1}{V_{k}}\right)\right)$, for the sequence of policies $\gamma_{k}$ with $V_{k} \downarrow 0$.
Let $\gamma'_{k}$ be a sequence of $\epsilon$-optimal policies for TRADEOFF for the sequence $P_{c,k}$.
Then we have that $\overline{P}(\gamma'_{k}) \downarrow c(\lambda)$ and $\overline{Q}(\gamma'_{k}) = \Omega\brap{\log\nfrac{1}{{P_{c,k} - c(\lambda)}}}$.
Since $\gamma'_{k}$ is $\epsilon$-optimal, we have that $Q^*(P_{c,k}) \geq \overline{Q}(\gamma'_{k}) - \epsilon$.
Therefore, $Q^*(P_{c,k}) = \Omega\brap{\log\nfrac{1}{{P_{c,k} - c(\lambda)}}}$ as $P_{c,k} \downarrow c_{R}(\lambda)$.\qeda

\subsection{Proof of Proposition \ref{lemma:case3lb}}
\label{appendix:proof:case3lb}

We consider a particular policy $\gamma$ in the sequence $\gamma_{k}$ with $V_{k} = V$.
Let $q_{d} \Deq \sup\brac{q : \sq \leq \lambda + \epsilon_{V}}$, where $\epsilon_{V} > 0$ is chosen as a function of $V$ in the following.
We note that as $\overline{s}(0) = 0$, the above set is non-empty.
Suppose $q_{d}$ is finite.

We note that by the admissibility of $\gamma$, $\forall q \in \brac{0,\dots, q_{d}}$, $\sq \leq \lambda + \epsilon_{V}$.
Hence, using $d = \epsilon_{V}$, we have from Proposition \ref{prelimresults:drift_real}, for a $\bar{q} \leq q_{d}$:
\begin{eqnarray*}
  & Pr\brac{Q \geq \bar{q}} \geq \brap{\frac{\epsilon_{a}}{\epsilon_{a} + \epsilon_{V}}}^{\bar{q}} \\
  & + \brap{1 - \brap{\frac{\epsilon_{a}}{\epsilon_{a} + \epsilon_{V}}}^{\bar{q}}}\bigg [Pr\brac{Q \geq q_{d} + 1} + \\
    & \frac{1}{\epsilon_{V}} \sum_{q = q_{d} + 1}^{\infty} \pi(q) \bras{\lambda - \sq}\bigg].
\end{eqnarray*}
Or 
\begin{eqnarray}
  & Pr\brac{Q < \bar{q}} \leq 1 - \brap{\frac{\epsilon_{a}}{\epsilon_{a} + \epsilon_{V}}}^{\bar{q}} - \brap{1 - \brap{\frac{\epsilon_{a}}{\epsilon_{a} + \epsilon_{V}}}^{\bar{q}}} \times \nonumber \\
  & \brap{\frac{1}{\epsilon_{V}}\sum_{q =  q_{d} + 1}^{\infty} \bras{\lambda - \sq} \pi(q)},
  \label{eq:case2lb1}
\end{eqnarray}
as $Pr\brac{Q \geq q_{d} + 1} \geq 0$.
For brevity, let $D_{t} \Deq -\brap{\frac{1}{\epsilon_{V}}\sum_{q =  q_{d} + 1}^{\infty} \bras{\lambda - \sq} \pi(q)}$.
We note that $D_{t}$ is positive, as for $q \geq q_{d} + 1$, $\lambda - \sq < -\epsilon_{V}$.

We recall that $c(s)$ is piecewise linear.
Let $\lambda = a_{p}$ for some $p > 1$.
Let $m$ be the tangent of the angle between (i) the line passing through $(a_{p + 1}, c(a_{p + 1}))$ and  $(\lambda, c(\lambda))$, and (ii) $l(s)$.
Then $m \sum_{q = q_{d} + 1}^{\infty} \pi(q) \brap{\sq - \lambda} \leq \sum_{q = q_{d} + 1}^{\infty} \pi(q) \bras{c(\sq) - l(\sq)}$.
Furthermore from the convexity of $c(s)$, linearity of $l(s)$, and as $c(s) - l(s) \geq 0$, we have that 
\begin{eqnarray*}
  & \sum_{q = q_{d} + 1}^{\infty} \pi(q) \bras{c(\sq) - l(\sq)} \leq \Expp \bras{c(\sQ) - l(\sQ)} \\
  & \leq \Exp \bras{c(S(Q)) - l(S(Q))} \leq V.
\end{eqnarray*}
Therefore \[ D_{t}\ \leq\ \frac{V}{m\epsilon_{V}}. \]

Now, we find a lower bound $\frac{\bar{q}}{2}$ on $\overline{Q}(\gamma)$ by finding the largest $\bar{q}$ such that $Pr\brac{Q < \bar{q}} \leq \frac{1}{2}$. 
A lower bound $\bar{q}_{1}$ to $\bar{q}$ can be obtained by using the upper bound \eqref{eq:case2lb1} on $Pr\brac{Q \leq \bar{q}}$.
Let $\bar{q}_{1}$ be the largest integer, if one exists, such that
\begin{eqnarray*}
  & 1 - \brap{\frac{\epsilon_{a}}{\epsilon_{a} + \epsilon_{V}}}^{\bar{q}_{1}} - \brap{1 - \brap{\frac{\epsilon_{a}}{\epsilon_{a} + \epsilon_{V}}}^{\bar{q}_{1}}}\times\\
  & \brap{\frac{1}{\epsilon_{V}}\sum_{q =  q_{d} + 1}^{\infty} \bras{\lambda - \sq} \pi(q)} \leq \frac{1}{2}.
\end{eqnarray*}
Then $\bar{q}_{1} \leq \bar{q}$.
Then we have to find $\bar{q}_{1}$ such that 
\begin{eqnarray*}
  1 - \brap{\frac{\epsilon_{a}}{\epsilon_{a} + \epsilon_{V}}}^{\bar{q}_{1}} + \brap{1 - \brap{\frac{\epsilon_{a}}{\epsilon_{a} + \epsilon_{V}}}^{\bar{q}_{1}}}D_{t} \leq \frac{1}{2}, \\
  \text{or } \frac{1 + 2D_{t}}{2 + 2D_{t}} \leq \nfrac{\epsilon_{a}}{\epsilon_{a} + \epsilon_{V}}^{\bar{q}_{1}},\\
  \text{or } \brap{1 + \frac{\epsilon_{V}}{\epsilon_{a}}}^{\bar{q}_{1}} \leq \frac{2 + 2D_{t}}{1 + 2D_{t}}
\end{eqnarray*}
We note that if $q_{d} = \infty$, then $D_{t} = 0$.
However, $\bar{q}_{1}$ satisfying the above inequality for finite $q_{d}$ is a lower bound for $\bar{q}_{1}$ for $q_{d} = \infty$.
Hence, we proceed with finding the above $\bar{q}_{1}$.
Let $\bar{q}_{2}$ be the largest integer such that 
\begin{eqnarray}
  \brap{1 + \frac{\epsilon_{V}}{\epsilon_{a}}}^{\bar{q}_{2}} \leq \frac{2}{1 + 2D_{t}}.
  \label{prop:case31}
\end{eqnarray}
Then $\bar{q}_{2} \leq \bar{q}_{1}$.
From \eqref{prop:case31} and the upper bound $\frac{V}{m\epsilon_{V}}$ on $D_{t}$, if $\bar{q}_{3}$ is the largest integer such that 
\begin{eqnarray*}
  & & \brap{1 + \frac{\epsilon_{V}}{\epsilon_{a}}}^{\bar{q}_{3}} \leq \frac{2}{1 + 2\frac{V}{m\epsilon_{V}}}, \\
\end{eqnarray*}
then $\bar{q}_{3} \leq \bar{q}_{2}$.
Or, we have that $\bar{q}_{3}$ is the largest integer such that
\begin{eqnarray*}
  & & \bar{q}_{3} \leq \log_{\brap{1 + \frac{\epsilon_{V}}{\epsilon_{a}}}} \brap{ \frac{2}{1 + \frac{2V}{m\epsilon_{V}}}}.
\end{eqnarray*}
We note that, as $V \downarrow 0$, if $\frac{V}{\epsilon_{V}} \rightarrow \infty$, then the bound will be negative.
We choose $\epsilon_{V} = aV$, where $a > \frac{2}{m}$.
Then we obtain that 
\begin{eqnarray*}
  & & \bar{q}_{3} \leq \log_{\brap{1 + \frac{a{V}}{\epsilon_{a}}}} \brap{ \frac{2}{1 + \frac{2}{ma}}},
\end{eqnarray*}
where the RHS is positive as $V \downarrow 0$.
Therefore the maximum $\bar{q}_{3}$ is atleast
\begin{eqnarray*}
  & & \floor{\log_{\brap{1 + \frac{a{V}}{\epsilon_{a}}}} \brap{ \frac{2}{1 + \frac{2}{ma}}}}.
\end{eqnarray*}
Since $\overline{Q}(\gamma) \geq \frac{\bar{q}}{2} \geq \frac{\bar{q}_{1}}{2}  \geq \frac{\bar{q}_{2}}{2}  \geq \frac{\bar{q}_{3}}{2}$ and $\log \brap{1 + \frac{a{V}}{\epsilon_{a}}} = \Theta\brap{V}$, we have that for the sequence of policies $\gamma_{k}$, $\overline{Q}(\gamma_{k}) = \Omega\nfrac{1}{V_{k}}$.
Let $\gamma'_{k}$ be a sequence of $\epsilon$-optimal policies for TRADEOFF for the sequence $P_{c,k}$.
Then we have that $\overline{P}(\gamma'_{k}) \downarrow c(\lambda)$ and $\overline{Q}(\gamma'_{k}) = \Omega\nfrac{1}{{P_{c,k} - c(\lambda)}}$.
Since $\gamma'_{k}$ is $\epsilon$-optimal, we have that $Q^*(P_{c,k}) \geq \overline{Q}(\gamma'_{k}) - \epsilon$.
Therefore, $Q^*(P_{c,k}) = \Omega\nfrac{1}{{P_{c,k} - c(\lambda)}}$ as $P_{c,k} \downarrow c_{R}(\lambda)$.\qeda

\subsection{Proof of Proposition \ref{lemma:realvalued}}
\label{appendix:proof:realvalued}
We consider a particular policy $\gamma$ in the sequence $\gamma_{k}$ with $V_{k} = V$.
Let $q_{d} = \sup\brac{ q : \sq \leq \lambda + \epsilon_{V}}$, where $\epsilon_{V}$ is chosen as a function of $V$ in the following.
Suppose $q_{d}$ is finite.
From the admissibility of $\gamma$, we have that $\forall q \in [0, q_{d}], \sq \leq \lambda + \epsilon_{V}$.
Using $d = \epsilon_{V}$ in Proposition \ref{prelimresults:drift_real}, we have for a $\bar{q} = k\Delta \leq q_{d}$, $k \geq 0$,
\begin{eqnarray*}
  Pr\brac{Q \geq \bar{q}} \geq \brap{\dep}^{k} + \brap{1 - \brap{\dep}^{k}} \times \\ 
  \bras{Pr\brac{Q > q_{d}} - \frac{1}{\epsilon_{V}}\int_{q_{d}^{+}}^{\infty} (\sq - \lambda)d\pi(q)}.
\end{eqnarray*}
Or we have that
\begin{eqnarray}
  & Pr\brac{Q < \bar{q}} \leq \brap{1 - \brap{\dep}^{k}} \times \nonumber \\
  & \bras{1 + \frac{1}{\epsilon_{V}} \int_{q_{d}^{+}}^{\infty} (\sq - \lambda)d\pi(q)},
  \label{eq:realvalevol}
\end{eqnarray}
as $Pr\brac{Q > q_{d}} \geq 0$.
We note that for $q \in [q_{d},\infty)$, $\sq - \lambda \geq \epsilon_{V}$.
  For brevity, we denote $\frac{1}{\epsilon_{V}}\int_{q_{d}^{+}}^{\infty} (\sq - \lambda)d\pi(q)$ by $D_{t}$. 
  We note that $D_{t}$ is positive.
  Now we note that for the policy $\gamma$, $\Exp c_{R}(S(Q, H)) - c_{R}(\lambda) \leq V$.
  Define $l(s)$ as the tangent to the curve $c_{R}(s)$ at $(\lambda, c_{R}(\lambda))$.
  Then we have that $\Exp [c_{R}(S(Q, H)) - l(S(Q, H))] \leq V$.
  Now as $c_{R}(s)$ is convex and $l(s)$ is linear, using Jensen's inequality we have that $\mathbb{E}_{\pi}[c_{R}(\sQ) - l(\sQ)] \leq V$.
  As in \cite[step (41)]{berry}, $c_{R}(s) - l(s) = G(s - \lambda)$ where $G(x)$ is a strictly convex function with $G(0) = 0$, $G'(0) = 0$, and $G''(0) > 0$.
  Thus we have that $\mathbb{E}_{\pi} G(\sQ - \lambda) \leq V$.
  Using the sequence of steps (45), (46), (47), and (48) of Berry and Gallager \cite{berry}, we obtain that
  \begin{eqnarray*}
    \bras{\int_{q_{d}^{+}}^{\infty} (\sq - \lambda) d\pi(q)}^{2} \leq \frac{V}{a_{1}},
  \end{eqnarray*}
  where $a_{1} > 0$ is such that $G(x) \geq a_{1} x^{2}$, for $x \in [-\lambda, S_{max} - \lambda]$.
  We note that then $D_{t} \leq \frac{1}{\epsilon_{V}} \sqrt{\frac{V}{a_{1}}}$.
  Choosing $\epsilon_{V} = 4\sqrt{\frac{V}{a_{1}}}$ we obtain that $D_{t} \leq \frac{1}{4}$.

  We note that $\overline{Q}(\gamma) \geq \frac{\bar{q}}{2}$, where $\bar{q} = \sup\brac{q : Pr\brac{Q < q} \leq \frac{1}{2}}$.
  Using the upper bound \eqref{eq:realvalevol}, if $\bar{q}_{1} = k_{1}\Delta$, where $k_{1}$ is the largest integer such that
  \begin{eqnarray*}
    \brap{1 - \brap{\dep}^{k_{1}}}\bras{1 + D_{t}} \leq \frac{1}{2},
  \end{eqnarray*}
  then $\bar{q}_{1} \leq \bar{q}$.
  We have that 
  \begin{eqnarray*}
    \frac{1 + 2D_{t}}{2 + 2D_{t}} \leq \nfrac{\delta\epsilon_{a}}{\delta\epsilon_{a} + \epsilon_{V}}^{k_{1}},\\
    \brap{1 + \frac{\epsilon_{V}}{\delta\epsilon_{a}}}^{k_{1}} \leq \frac{2 + 2D_{t}}{1 + 2D_{t}}
  \end{eqnarray*}
  Let $k_{2}$ be the largest integer such that 
  \begin{eqnarray}
    \brap{1 + \frac{\epsilon_{V}}{\delta\epsilon_{a}}}^{k_{2}} \leq \frac{2}{1 + 2D_{t}}.
    \label{eq:real5}
  \end{eqnarray}
  Then $k_{2} \leq k_{1}$.
  We note that even if $q_{d}$ is infinite, $ k_{2} \Delta$ is a lower bound to $\bar{q}$, since $D_{t}$ is $0$ in that case.
  The rest of the proof holds irrespective of whether $q_{d}$ is finite or infinite.

  Then, from \eqref{eq:real5} and using the upper bound $\frac{1}{4}$ on $D_{t}$, if $k_{3}$ is the largest integer such that 
  \begin{eqnarray*}
    \brap{1 + \frac{\epsilon_{V}}{\delta\epsilon_{a}}}^{k_{3}} \leq \frac{2}{1 + \frac{1}{2}},
  \end{eqnarray*}
  then $k_{3} \leq k_{2}$.
  We obtain that $k_{3}$ is atleast 
  \begin{eqnarray*}
    \log_{\brap{1 + \frac{\epsilon_{V}}{\delta\epsilon_{a}}}} \brap{\frac{4}{3}} - 1.
  \end{eqnarray*}
  Since $\overline{Q}(\gamma) \geq \frac{\bar{q}}{2} \geq \frac{\Delta k_{1}}{2} \geq \frac{\Delta k_{2}}{2} \geq \frac{\Delta k_{3}}{2}$, we have that $\overline{Q}(\gamma) \geq \frac{\Delta}{2} \brap{\log_{\brap{1 + \frac{\epsilon_{V}}{\delta\epsilon_{a}}}} \brap{\frac{4}{3}} - 1}$.
  Since $\log\brap{1 + \frac{\epsilon_{V}}{\delta\epsilon_{a}}} = \Theta\brap{\sqrt{V}}$, we have that for the sequence $\gamma_{k}$ as $V_{k} \downarrow 0$, $\overline{Q}(\gamma_{k}) = \Omega\nfrac{1}{\sqrt{V_{k}}}$.

  Let $\gamma'_{k}$ be a sequence of $\epsilon$-optimal policies for TRADEOFF for the sequence $P_{c,k}$.
  Then we have that $\overline{P}(\gamma'_{k}) \downarrow c_{R}(\lambda)$ and $\overline{Q}(\gamma'_{k}) = \Omega\nfrac{1}{\sqrt{P_{c,k} - c_{R}(\lambda)}}$.
  Since $\gamma'_{k}$ is $\epsilon$-optimal, we have that $Q^*(P_{c,k}) \geq \overline{Q}(\gamma'_{k}) - \epsilon$.
  Therefore, $Q^*(P_{c,k}) = \Omega\nfrac{1}{\sqrt{P_{c,k} - c(\lambda)}}$ as $P_{c,k} \downarrow c_{R}(\lambda)$.\qeda

  \subsection{Outline of proof for Proposition \ref{prop:rmodel_piecewisecost}}
  \label{appendix:proof:rmodel_piecewisecost}
  We consider Case 2 first.
  We define $s_{l}, s_{u}$, and the line $l(s)$ as in Section \ref{sec:differentcases}.
  Consider any sequence of admissible policies $\gamma_{k}$ with $\Pgk - c_{R}(\lambda) = V_{k} \downarrow 0$.
  Then we have that $\Exp c_{R}(\sQ) - c_{R}(\lambda) \leq V_{k}$.
  For a particular policy $\gamma$ in the sequence, $0 < \epsilon < s_{l}$, and $q_{1} \Deq s_{l} - \epsilon$ we have that $Pr\brac{Q < q_{1}} \leq \frac{V}{m\epsilon}$ as in the proof of Proposition \ref{lemma:case2}.
  We note that for the R-model, the queue evolution is on $\sR$.
  We discretize $\sR$ into a countable number of intervals $\brap{[0, \Delta), [\Delta, 2 \Delta), \dots}$.
  Using the upper bound on the stationary probability distribution from Proposition \ref{prelimresults:txprob_real}, and proceeding as in the proof of Proposition \ref{lemma:case2}, we can show that $\Qgk = \Omega\brap{\log\nfrac{1}{V_{k}}}$.
  A complete illustration of this proof technique is given in the proof of Proposition \ref{lemma:tradeoffutilitylb}.
  The asymptotic lower bound on $Q^*(P_{c})$ can be derived as in the proof of Proposition \ref{lemma:case2}.

  For Case 3, we proceed as in the proof of Proposition \ref{lemma:realvalued} by defining $q_{d}$ to be $\sup\brac{q : \sq \leq \lambda + \epsilon_{V}}$, where $\epsilon_{V}$ is a function of $V$ to be chosen in the following.
  We note that $\lambda = a_{p}$ for some $p > 1$.
  We recall that $D_{t} \Deq \frac{1}{\epsilon_{V}} \int_{q_{d}}^{\infty} \brap{\sq - \lambda} d\pi(q)$.
  For the policy $\gamma$ we have that
  \begin{eqnarray*}
    \Exp \bras{c_{R}(\sQ) - l(\sQ)} & \leq & V, \\
    \int_{q_{d}}^{\infty} \brap{c_{R}(\sq) - l(\sq)} d\pi(q) & \leq & V, \text{ or,} \\
    \frac{1}{\epsilon_{V}} \int_{q_{d}}^{\infty} \brap{\sq - \lambda} d\pi(q) & \leq & \frac{V}{m\epsilon_{V}},
  \end{eqnarray*}
  where $m$ is the tangent of angle made by the line passing through $( a_{p - 1}, c_{R}(a_{p - 1}))$ and $(\lambda, c_{R}(\lambda))$ with $l(s)$.
  Now we choose $\epsilon_{V} = \frac{4V}{m}$ to obtain that $D_{t} \leq \frac{1}{4}$.
  Then we proceed as in the proof of Proposition \ref{lemma:realvalued} to obtain that $\Qgk = \Omega\nfrac{1}{V_{k}}$.
  The asymptotic lower bound on $Q^*(P_{c})$ can be derived as in the proof of Proposition \ref{lemma:case3lb}.
  
  \subsection{Proof of Proposition \ref{lemma:tradeoffutilitylb}}
  \label{appendix:proof:tradeoffutilitylb}
  For the policy $\gamma$, let $q_{1} \stackrel{\Delta} = \sup\brac{q : \sq \leq \rho\lambda - \epsilon_{1}}$.
  Let $\Delta > 0$ and $\epsilon_{1} > 0$ be chosen such that $0 < \epsilon_{1} < \rho\lambda - \Delta$.
  We have that $\Expp c_{R}(\sQ) = \Expp \bras{ c_{R}(\rho\lambda) + \frac{dc_{R}(x)}{dx} \vert_{\rho\lambda} (\sQ - \rho\lambda) + G(\sQ - \rho\lambda)}$,
  where $G(x)$ is a strictly convex function as in \cite[eq (41)]{berry}, with $G(0) = 0$, $G'(0) = 0$, and $G''(0) > 0$.
  Since $\Sg = \Expp \sQ \geq \rho\lambda$ we have that
  \begin{eqnarray*}
    \Expp c_{R}(\sQ) - c_{R}(\rho\lambda) \geq \Expp G(\sQ - \rho\lambda).
  \end{eqnarray*}
  Thus, we have that $\Expp G(\sQ - \rho\lambda) \leq V$.
  Therefore, for $q_{1}$ as defined above, proceeding as in the proof of Lemma \ref{lemma:prelimresults:serviceratebound}, we have that there exists a positive $a_{1}$ such that
  \begin{eqnarray}
    Pr\brac{Q < q_{1}} \leq \frac{V}{a_{1} \epsilon_{1}^{2}}.
    \label{eq:q1upperb}
  \end{eqnarray}
  From Proposition \ref{prelimresults:txprob_real}, we have that for $m \geq 0$, 
  \begin{eqnarray*}
    \pi\bigg[q_{1} + m\Delta, q_{1} + (m + 1)\Delta\bigg) \leq \frac{\pi[0, q_{1})}{\rho_{d}}\brap{1 + \frac{1}{\rho_{d}}}^{m},
  \end{eqnarray*}
  where $\rho_{d} = Pr\brac{A[1] = 0}\frac{\rho\lambda - \epsilon_{1} - \Delta}{S_{max} - \Delta}$.
  Since $Pr\brac{A[1] = 0} \geq Pr\brac{R[1] = 0}$, with $\rho_{a} = Pr\brac{R[1] = 0}\frac{\rho\lambda - \epsilon_{1} - \Delta}{S_{max} - \Delta}$, we have that 
  \begin{eqnarray*}
    \pi\bigg[q_{1} + m\Delta, q_{1} + (m + 1)\Delta\bigg) \leq \frac{\pi[0, q_{1})}{\rho_{a}}\brap{1 + \frac{1}{\rho_{a}}}^{m},
  \end{eqnarray*}
  Hence, we have that for $m \geq 0$,
  \begin{eqnarray}
    \pi\left[q_{1}, q_{1} + m \Delta\right) & \leq & \frac{\pi[0, q_{1})}{\rho_{a}} \frac{\brap{1 + \frac{1}{\rho_{a}}}^{m} - 1}{1 + \frac{1}{\rho_{a}} - 1} \nonumber \\
        & = & \pi[0, q_{1}) \bras{\brap{1 + \frac{1}{\rho_{a}}}^{m} - 1}.
          \label{eq:probboundonq1toq1m}
  \end{eqnarray}
  Since $Pr\brac{Q < q_{1}} \leq \frac{V}{a_{1} \epsilon^{2}}$ (from \eqref{eq:q1upperb}), if $m$ is the largest integer such that  
  \begin{eqnarray*}
    \pi\left[0, q_{1}\right) + \pi\left[q_{1}, q_{1} + m \Delta\right) \leq \frac{1}{2},
  \end{eqnarray*}
  then $\overline{Q}(\gamma) \geq \frac{m\Delta}{2}$.
  From \eqref{eq:probboundonq1toq1m}, we have that
  \begin{eqnarray*}
    \pi\left[0, q_{1}\right) + \pi[0, q_{1}) \bras{\brap{1 + \frac{1}{\rho_{a}}}^{m_{1}} - 1} = \pi\left[0, q_{1}\right)\brap{1 + \frac{1}{\rho_{a}}}^{m_{1}}
  \end{eqnarray*}
  Suppose $m_{1}$ is the largest integer such that 
  \begin{eqnarray}
    \pi\left[0, q_{1}\right)\brap{1 + \frac{1}{\rho_{a}}}^{m_{1}} \leq \frac{1}{2}.
      \label{eq:lbut0}
  \end{eqnarray}
  Then $m_{1} \leq m$.
  Using \eqref{eq:q1upperb}, if $m_{2}$ is the largest integer such that
  \begin{eqnarray*}
    \brap{1 + \frac{1}{\rho_{a}}}^{m_{2}} & \leq & \frac{a_{1}\epsilon_{1}^{2}}{2V}, \text{ or }, \\
    m_{2} & \leq & \log_{\brap{1 + \frac{1}{\rho_{a}}}} \brap{\frac{a_{1}\epsilon_{1}^{2}}{2V}},
  \end{eqnarray*}
  then $m_{2} \leq m_{1}$.
  We have that 
  \[ m_{2} = \floor{\log_{\brap{1 + \frac{1}{\rho_{a}}}} \brap{\frac{a_{1}\epsilon_{1}^{2}}{2V}}}. \]
  Since $\overline{Q}(\gamma) \geq \frac{m\Delta}{2} \geq \frac{m_{1}\Delta}{2} \geq \frac{m_{2}\Delta}{2}$, we obtain that 
  \[ \overline{Q}(\gamma) \geq \frac{\Delta}{2} \brap{\log_{\brap{1 + \frac{1}{\rho_{a}}}} \brap{\frac{a_{1}\epsilon_{1}^{2}}{2V}} - 1}. \]
  So for the sequence of policies $\gamma_{k}$ with $V_{k} \downarrow 0$ we have that $\overline{Q}(\gamma_{k}) = \Omega\brap{\log\nfrac{1}{V_{k}}}$.
  Let $\gamma'_{k}$ be a sequence of $\epsilon$-optimal policies for TRADEOFF for the sequence $P_{c,k}$.
  Then we have that $\overline{P}(\gamma'_{k}) \downarrow c_{R}(\rho\lambda)$ and $\overline{Q}(\gamma'_{k}) = \Omega\brap{\log\nfrac{1}{{P_{c,k} - c_{R}(\lambda)}}}$.
  Since $\gamma'_{k}$ is $\epsilon$-optimal, we have that $Q^*(P_{c,k},\rho) \geq \overline{Q}(\gamma'_{k}) - \epsilon$.
  Therefore, $Q^*(P_{c,k}) = \Omega\brap{\log\nfrac{1}{{P_{c,k} - c_{R}(\lambda)}}}$ as $P_{c,k} \downarrow c_{R}(\lambda)$.\qeda

\end{appendices}

\end{document}